\documentclass[journal,twoside,web]{ieeeconf}
\usepackage{graphicx}
\usepackage{amsmath}
\usepackage{amsfonts}
\usepackage{hyperref}
\usepackage{booktabs}
\usepackage{siunitx}
\usepackage{algpseudocode}
\usepackage{amssymb}
\usepackage{dsfont}
\usepackage{subfiles}
\usepackage{algorithm}
\hypersetup{
    colorlinks=true,
    linkcolor=blue,
    filecolor=magenta,      
    urlcolor=cyan,
    pdftitle={Overleaf Example},
    pdfpagemode=FullScreen,
    }
\newtheorem{theorem}{Theorem}[section]
\newtheorem{lemma}{Lemma}[section]
\newtheorem{remark}{Remark}[section]
\newtheorem{definition}{Definition}[section]
\newtheorem{proposition}{Proposition}[section]
\newtheorem{corollary}{Corollary}[section]
\newcommand\E{\mathbb{E}}

\newcommand\Lcal{\mathcal{L}}

\newcommand\Jcal{\mathcal{J}}

\newcommand\Hcal{\mathcal{H}}
\newcommand\Xcal{\mathcal{X}}

\newcommand\norm[1]{\lVert#1\rVert}

\DeclareMathOperator*{\sym}{\mathbb{S}}

\newcommand{\R}{{\mathbb R}}  

\algdef{SE}[DOWHILE]{Do}{doWhile}{\algorithmicdo}[1]{\algorithmicwhile\ #1}
\usepackage{cite}
\usepackage{algorithmicx}
\usepackage{textcomp}
\begin{document}
\title{Parallel Policy-Gradient Method for Parameter Optimization of Nonlinear Feedback Controllers}
\author{An Nguyen and Leilei Cui
\thanks{This work is supported by the UNM School of Engineering faculty startup funds, the UNM Research Allocations Committee (RAC) grant, and the National Science Foundation under Award No. 2624522.}
\thanks{An Nguyen and Leilei Cui are with the Learning and Control (LC) lab, Department of Mechanical Engineering, University of New Mexico, Albuquerque, NM 87110 USA (email: annguyen2303@unm.edu; lcui@unm.edu).}}

\maketitle

\begin{abstract}
Structured feedback controllers provide stability guarantees, but their performance depends on manually tuned parameters. Policy-gradient methods offer a systematic way to tune these parameters, but evaluating the gradient is computationally expensive because it requires a sequential forward rollout of the state dynamics followed by a backward rollout of the costate dynamics. This letter develops a time-parallel policy-gradient framework for discrete-time nonlinear control-affine systems. We evaluate the policy gradient by formulating the forward state and backward costate rollouts as residual-minimization problems and solving them with Gauss–Newton (GN) iterations that use parallel associative scans. We show that global asymptotic stability and local exponential stability of the closed loop guarantee convergence of the parallel state solver without requiring bounds on Jacobian products throughout the state space. Near the exact trajectory, GN iterations converge quadratically in a neighborhood whose size is independent of the horizon \(T\). When initialized at the equilibrium trajectory, a clipped variant converges at a linear rate that is also independent of \(T\). We also quantify how horizon truncation and solver tolerance affect the accuracy of the policy gradient. Experiments on an LQR problem and a lossless Lorenz-96 system with saturated passivity-based control demonstrate improved closed-loop performance and computational benefits.
\end{abstract}

\begin{keywords}
Control-affine systems, nonlinear control, optimal control, policy optimization, parallel computing, policy gradient.
\end{keywords}

\section{Introduction}
\label{sec:introduction}
Nonlinear control methods are designed around analytical structure of the system \cite{khalil2002nonlinear,krstic1995nonlinear}, where they can provide stability guarantees while retaining physically meaningful controller structures. Their practical performance, however, often depends strongly on the selection of gains and other design parameters. These parameters are commonly chosen through trial and error or analytical guidelines, which becomes increasingly difficult for nonlinear systems with multiple coupled parameters. Policy-gradient methods offer an appealing alternative: rather than replacing the controller structure, they can optimize its parameters directly with respect to a performance objective \cite{massaroli2022optimal,larby2025optimal}. When optimization is restricted to a range of stabilizing controller parameters, the stability guarantees and physical structure of the original controller can be preserved. The success of policy-gradient methods in reinforcement learning and control further motivates their use as a systematic approach for tuning such structured feedback controllers \cite{hu2023toward,fazel2018global}.

For linear-quadratic control, policy gradients admit explicit forms with well-understood convergence and robustness properties \cite{fazel2018global,cui2024small}, whereas nonlinear systems generally require numerical gradient computation \cite{jin2020pontryagin}. Conventionally, computing the policy gradient requires sequential forward rollout of the state trajectory and backward propagation of the costate over the horizon $T$. This inherently sequential procedure can become computationally expensive for long-horizon problems and constitutes a major bottleneck of policy-gradient methods. In this paper, we develop a time-parallel approach for evaluating the state and costate trajectories required for policy-gradient computation, enabling substantial acceleration over conventional sequential methods on parallel hardware such as GPUs. Recent work has explored time-parallel evaluation of nonlinear dynamical systems. DEER \cite{lim2024parallelizing,blelloch1990prefix} reformulates sequential trajectory evaluation as an optimization problem involving a residual-based merit function and applies a Gauss-Newton (GN) method, with the resulting computations parallelized using associative scans. DeepPCR \cite{danieli2023deeppcr} adopts a related formulation but employs parallel cyclic reduction to exploit temporal parallelism. 
 
The computational efficiency of this reformulation depends on both the cost of each optimization step and the number of steps required for convergence. For a horizon of length $T$, associative scans reduce the computational depth of each GN step from \(\mathcal{O}(T)\) for a sequential rollout to \(\mathcal{O}(\log T)\) under parallel execution \cite{lim2024parallelizing,blelloch1990prefix}. The number of GN iterations depends on the initialization and the conditioning of the merit function. Existing analyses use bounds on products of system Jacobians to link contraction or predictability of the dynamics to both the conditioning of the merit function and the convergence rate of the GN solver \cite{gonzalez2026predictability,kozachkov2026contracting}. 
The analysis in \cite{gonzalez2026predictability} assumes a uniform exponential bound on products of
Jacobians at all states. Under this assumption, the authors bound the Polyak-\L{}ojasiewicz (PL) constant of the merit function and give an explicit basin of quadratic convergence of the GN iterates. They also establish global linear convergence of the GN iteration, although the prefactor in their Euclidean error bound may grow exponentially with the horizon \(T\) \cite[Theorem~4]{gonzalez2026predictability}. The analysis in \cite{kozachkov2026contracting} assumes contraction and bounds this constant for such systems. Moreover, the GN iteration is shown to terminate within $T$ iterations \cite{gonzalez2024towards}. However, a bound on products of Jacobians at all states can be hard to verify, and contraction is a strict condition that many closed-loop systems do not satisfy.
 
In this paper, we show that global asymptotic stability
together with local exponential stability yields a horizon-independent bound on the inverse residual Jacobian
near the exact trajectory, uniformly over a compact set of
initial states and of controller parameters. The key observation is that global asymptotic stability ensures entry into a sufficiently small neighborhood of the
equilibrium after a uniformly bounded number of steps,
while local exponential stability provides contraction
of the Jacobians there in an appropriate norm.
Consequently, only a uniformly bounded initial portion
of the trajectory may be noncontractive. Smoothness ensures
that the resulting Jacobian-product bounds also hold for
nearby trajectory estimates. This establishes a local PL inequality and local linear and
quadratic convergence of the state solver, with constants
and a neighborhood radius independent of the horizon.
 
The main contributions of this letter are:
\begin{enumerate}
\item
We develop a time-parallel method for optimizing nonlinear
feedback-controller parameters. The policy gradient is written in terms of forward state and backward costate recursions, and recovers the standard LQR policy gradient \cite{fazel2018global}. Both recursions are evaluated using Gauss-Newton updates
and associative scans that are amenable to GPU parallelization, and the costate pass requires a single scan once the states are fixed.
 
\item
In contrast to \cite{gonzalez2026predictability}, we do not assume bounds on products of Jacobians. We derive them near the exact trajectory from closed-loop global asymptotic stability and local
exponential stability, and obtain local inverse-Jacobian and PL bounds, yielding local linear and quadratic
convergence with horizon-independent constants and basin
radius. All constants are uniform over compact sets of initial states and of controller parameters.
 
\item
We prove that a clipped GN iteration initialized at the equilibrium converges globally at a
linear rate that is independent of the horizon, and we bound
the policy-gradient errors caused by truncating the
infinite horizon and by terminating the solver at a finite
tolerance.
\end{enumerate}
 
The proposed approach is evaluated on discrete-time LQR and a lossless Lorenz-96 system with saturated interconnection and damping assignment passivity-based control (IDA-PBC).

\section{Policy Optimization for Nonlinear Systems}
In this paper, we consider the discrete-time nonlinear system
\begin{align*}
    x_{k+1} = f(x_k) + g(x_k) u_k, \qquad k=0,1,2,\ldots
\end{align*}
where $x_k\in \R^n$ denotes the state, $u_k\in\R^m$ denotes the control input, and $f:\R^n\to\R^n$ and $g:\R^n\to\R^{n\times m}$ are assumed to be smooth, with $f(0)=0$. The feedback policy is parameterized as $u_k = \pi(x_k,\theta)$, where $\pi:\R^n\times\R^p \to \R^m$ is a smooth function with $\pi(0,\theta)=0$ and $\theta\in\R^p$ denotes the controller parameter vector.

Define the discrete-time trajectory as $x(k,x_0,\theta)$, which is generated by the closed-loop dynamics starting from $x_0$ under the policy $u_k = \pi(x_k,\theta)$. The cost functional is defined as
\begin{align*}
    \Jcal(x_0,\theta) &= \sum_{k=0}^\infty l(x(k,x_0,\theta), \theta), 
\end{align*}
where 
\[
l(x,\theta) = q(x) + \pi(x,\theta)^\top R \pi(x,\theta),
\]
$q:\R^n\to\R_+$ is a smooth positive-definite state-cost function and $R\in\sym^m_{++}$ is the control cost matrix. For all $k$, the closed-loop dynamics is given by
\begin{align}
    x_{k+1} := F(x_k,\theta) = f(x_k) + g(x_k) \pi(x_k,\theta).
    \label{eq:closedloop}
\end{align}
\begin{definition}
    A policy $\pi(x_k,\theta)$ is admissible if the closed-loop system \eqref{eq:closedloop} is globally asymptotically stable (GAS) and locally exponentially stable (LES) at the origin.
\end{definition}

We aim to develop a gradient-based method for solving
\begin{align*}
\min_{\theta\in\Theta}\;\bar{\Jcal}(\theta),
\qquad
\bar{\Jcal}(\theta)
:=
\E_{x_0\sim\mathcal{X}}
\left[\Jcal(x_0,\theta)\right],
\end{align*}
where $\Theta\mathrel{\subset}\R^p$ is a compact convex set of parameters for which the policy is admissible, and $\mathcal{X}$ is a distribution for initial states supported on a compact set $\Omega\subset\R^n$. Since $\Theta$ is compact and convex, the Euclidean projection
\begin{equation*}
    \text{Proj}_\Theta(v):=\arg\min_{\theta\in\Theta}\norm{\theta-v},\quad v\in\R^p
\end{equation*}
is well defined.
The policy gradient $\nabla_\theta \Jcal(x_0,\theta)$ is derived in the following theorem.

\begin{theorem}
    \label{theorem:policy_gradient}
    Let $\pi(x,\theta)$ be an admissible policy. Consider the following discrete-time Hamiltonian:
    \begin{align*}
    \Hcal(x_k,\lambda_{k+1},\theta) = l(x_k,\theta) + \lambda_{k+1}^\top F(x_k,\theta),
    \end{align*}
    with the discrete-time Hamiltonian dynamics given by:
    \begin{subequations}\label{eq:HamiltonDynamics}
        \begin{align}
        x_{k+1} &= \frac{\partial\Hcal(x_k,\lambda_{k+1},\theta)}{\partial \lambda_{k+1}}, \qquad x_0 \in \Omega,\\
        \lambda_k &= \frac{\partial\Hcal(x_k,\lambda_{k+1},\theta)}{\partial x_k}, \qquad \lim_{k\to\infty} \lambda_k = 0.
        \end{align}
    \end{subequations}   
The policy gradient of the cost functional $\Jcal(x_0,\theta)$ is
    \begin{equation}
    \begin{split}
        &\nabla_\theta \Jcal(x_0,\theta) \\
        &= \sum_{k=0}^\infty \left(\nabla_\theta \pi(x_k,\theta)\right)^\top 
    \Big(2R\,\pi(x_k,\theta) + g(x_k)^\top \lambda_{k+1}\Big).
    \end{split}
    \label{eq:policy_gradient}
    \end{equation}
\end{theorem}
\begin{proof}
    Define the sensitivity function as 
    \begin{align*}
    S(k,x_0,\theta) = \frac{\partial x(k,x_0,\theta)}{\partial \theta}\in \R^{n\times p}.
    \end{align*}
    Taking the derivative of the closed-loop dynamics with respect to $\theta$ gives
    \begin{align}
    S(k+1,x_0,\theta) = \nabla_x F(x_k,\theta) S(k,x_0,\theta) + \nabla_\theta F(x_k,\theta).
    \label{eq:sk1}
    \end{align}
    The costate dynamics in \eqref{eq:HamiltonDynamics} can be rewritten as
    \begin{align}
    \lambda_k :=G(x_k,\lambda_{k+1},\theta) = \nabla_x l(x_k,\theta) + \nabla_x F(x_k,\theta)^\top \lambda_{k+1}.
    \label{eq:lamk1}
    \end{align}
    Hence, by \eqref{eq:sk1} and \eqref{eq:lamk1} we can compute
    \begin{align*}
    S(k+1,x_0,\theta)^\top \lambda_{k+1}&=\nabla_\theta F(x_k,\theta)^\top \lambda_{k+1} +  S(k,x_0,\theta)^\top \lambda_k\\
    &\qquad   - S(k,x_0,\theta)^\top \nabla_x l(x_k,\theta).
    \end{align*}
    Group the terms and define the forward difference:
    \begin{equation*}
    \Delta(S(k,x_0,\theta)^\top \lambda_k) = S(k+1,x_0,\theta)^\top \lambda_{k+1} - S(k,x_0,\theta)^\top \lambda_k.
    \end{equation*}
    Therefore, it follows that
    \begin{align*}
    &\Delta(S(k,x_0,\theta)^\top \lambda_k)\\
    &= - S(k,x_0,\theta)^\top \nabla_x l(x_k,\theta) + \nabla_\theta F(x_k,\theta)^\top \lambda_{k+1}.
    \end{align*}
    
    Additionally, the derivative of the cost functional with respect to $\theta$, is:
    \begin{align*}
    \nabla_\theta \Jcal(x_0,\theta) = \sum_{k=0}^\infty \Big( S(k,x_0,\theta)^\top \nabla_x l(x_k,\theta) + \nabla_\theta l(x_k,\theta) \Big).
    \end{align*}
    Substituting the expression of $\Delta(S(k,x_0,\theta)^\top \lambda_k)$ into the above equation gives
    \begin{align*}
    \nabla_\theta \Jcal(x_0,\theta)
    &= \sum_{k=0}^\infty \Big(\nabla_\theta F(x_k,\theta)^\top \lambda_{k+1}+ \nabla_\theta l(x_k,\theta) \Big)\\ 
    &- \lim_{k\to\infty}(S(k,x_0,\theta)^\top \lambda_k) + S(0,x_0,\theta)^\top \lambda_0.
    \end{align*}
    In the above equation, since $\lim_{k\to\infty} \lambda_k = 0$, $S(k,x_0,\theta)$ is bounded (see \eqref{eq:costate-sensitivity-decay}), and $S(0,x_0,\theta) = 0$, the boundary terms vanish. Hence, we have
    \begin{align*}
    \nabla_\theta \Jcal(x_0,\theta) &= \sum_{k=0}^\infty \Big(\nabla_\theta F(x_k,\theta)^\top \lambda_{k+1} + \nabla_\theta l(x_k,\theta) \Big).
    \end{align*}
    
    It follows from the closed-loop dynamic and the cost function that
    \begin{align*}
    \nabla_\theta F(x_k,\theta)^\top &= \nabla_\theta \pi(x_k,\theta)^\top g(x_k)^\top,\\
    \nabla_\theta l(x_k,\theta) &= 2\nabla_\theta \pi(x_k,\theta)^\top R\pi(x_k,\theta) ,
    \end{align*}
    which, in turn, results in \eqref{eq:policy_gradient}.
    
\end{proof}

\begin{remark}\label{remark:costate}
By the Bellman equation, the cost function satisfies
\begin{align*}
\Jcal(x_k,\theta)
=
l(x_k,\theta)
+
\Jcal(F(x_k,\theta),\theta).
\end{align*}
Differentiating both sides with respect to $x_k$ gives
\begin{align*}
\nabla_x \Jcal(x_k,\theta)
=
\nabla_x l(x_k,\theta)
+
\nabla_x F(x_k,\theta)^\top
\nabla_x \Jcal(x_{k+1},\theta).
\end{align*}
Therefore, by comparing with \eqref{eq:lamk1}, the costate admits the interpretation
\[
\lambda_k
=
\nabla_x \Jcal(x_k,\theta).
\]
\end{remark}


\begin{remark}
Define the state occupation measure induced by the policy $\pi(\cdot,\theta)$ and the initial-state distribution $\mathcal{X}$ as
\begin{align*}
    \rho_\theta(\mathrm{d}x)
    :=
    \E_{x_0\sim\mathcal{X}}
    \left[
        \sum_{k=0}^{\infty}
        \delta_{x(k,x_0,\theta)}(\mathrm{d}x)
    \right].
\end{align*}
Then, for any integrable function $\phi$,
\begin{align*}
    \int_{\R^n}\phi(x)\rho_\theta(\mathrm{d}x)
    =
    \E_{x_0\sim\mathcal{X}}
    \left[
        \sum_{k=0}^{\infty}\phi(x(k,x_0,\theta))
    \right].
\end{align*}
Using the $Q$-function
\begin{align*}
    \mathcal{Q}(x,u,\theta)
    :=
    q(x)+u^\top Ru
    +
    \Jcal\!\left(f(x)+g(x)u,\theta\right),
\end{align*}
we have
\begin{align*}
    \nabla_u \mathcal{Q}(x,u,\theta)
    =
    2Ru
    +
    g(x)^\top
    \nabla_x \Jcal\!\left(f(x)+g(x)u,\theta\right).
\end{align*}
Therefore, by Remark \ref{remark:costate}, the policy gradient in Theorem~\ref{theorem:policy_gradient} can be equivalently written as
\begin{align*}
    \nabla_\theta \bar{\Jcal}(\theta)
    =
    \int_{\R^n}
    \left(\nabla_\theta \pi(x,\theta)\right)^\top
    \left.
    \nabla_u \mathcal{Q}(x,u,\theta)
    \right|_{u=\pi(x,\theta)}
    \rho_\theta(\mathrm{d}x).
\end{align*}

This expression has the same form as the deterministic policy gradient theorem in \cite{silver2014deterministic}, with $\rho_\theta$ playing the role of the state visitation measure.
\end{remark}

For linear time-invariant (LTI) systems with quadratic costs, the policy-gradient expression in Theorem~\ref{theorem:policy_gradient} reduces to the standard LQR policy gradient derived in \cite{fazel2018global}.

\begin{corollary}
    Consider the LTI system with $f(x_k)=Ax_k$ and $g(x_k)=B$, where $A\in\R^{n\times n}$ and $B\in\R^{n\times m}$, together with the quadratic state cost \( q(x_k)=x_k^\top Qx_k\), where $Q\in\mathbb{S}_{++}^n$. Then, for any stabilizing linear
    policy $u_k=-Kx_k$, the policy gradient is given by
    \begin{equation*}
        \nabla_K \Jcal(x_0,K)
        =
        2\left[
        \left(R+B^\top P_KB\right)K
        -B^\top P_KA
        \right]\Sigma_K,
    \end{equation*}
    where \(\Sigma_K := \sum_{k=0}^{\infty}x_kx_k^\top\), and $P_K=P_K^\top\succ0$ is the solution to
    \begin{equation*}
        P_K
        =
        (A-BK)^\top P_K(A-BK)
        +Q+K^\top RK.
    \end{equation*}
\end{corollary}

\begin{proof}
    It follows from Theorem~\ref{theorem:policy_gradient} that, for the
    LQR problem,
    \begin{equation}
        \label{eq:lqr_policy_gradient}
        \nabla_K \Jcal(x_0,K)
        =
        \sum_{k=0}^{\infty}
        \left(
            2RKx_k-B^\top\lambda_{k+1}
        \right)x_k^\top.
    \end{equation}
    
    Since $K$ is stabilizing, the corresponding value function is quadratic \cite{Hewer1971}, i.e. \(\Jcal(x_k,K)=x_k^\top P_Kx_k\). Hence, by Remark~\ref{remark:costate},
    \[
        \lambda_{k+1}
        =
        \nabla_{x_{k+1}}\Jcal(x_{k+1},K)
        =
        2P_Kx_{k+1}
        =
        2P_K(A-BK)x_k.
    \]
    Substituting this expression into
    \eqref{eq:lqr_policy_gradient} yields
    \begin{align*}
        \nabla_K \Jcal(x_0,K)
        &=
        2\left[
            RK-B^\top P_K(A-BK)
        \right]
        \sum_{k=0}^{\infty}x_kx_k^\top \\
        &=
        2\left[
            \left(R+B^\top P_KB\right)K
            -B^\top P_KA
        \right]\Sigma_K,
    \end{align*}
    which completes the proof.
\end{proof}

\section{Parallelized Monte Carlo Approximation of the Policy Gradient}
In this section, we employ a Monte Carlo method to approximate $\nabla_\theta \bar{\Jcal}(\theta) = \E_{x_0\sim\Xcal} \left[\nabla_\theta \Jcal(x_0,\theta)\right]$. Specifically, we draw $N$ initial states $\{x_{0,j}\}_{j=1}^N$ independently and identically distributed (i.i.d.) from $\Xcal$, where $j$ is the index of the $j$-th sample, and approximate the policy gradient by
\begin{align}\label{eq:MCsampling}
    \hat{\nabla}_\theta \bar{\Jcal}(\theta)
    =
    \frac{1}{N}
    \sum_{j=1}^N
    \nabla_\theta \Jcal(x_{0,j},\theta).
\end{align}
 
Let $x(k,x_{0,j},\theta)$ denote the state trajectory of the closed loop system with $x_0=x_{0,j}$. We truncate \eqref{eq:policy_gradient} at a finite horizon $T$ and compute
\begin{equation}\label{eq:truncated_gradient}
\begin{split}
    &\nabla_\theta \Jcal_T(x_0,\theta) \\
    &= \sum_{k=0}^{T} \left(\nabla_\theta \pi(x_k,\theta)\right)^\top
    \Big(2R\,\pi(x_k,\theta) + g(x_k)^\top \lambda_{k+1}\Big),
\end{split}
\end{equation}
where $\lambda_k=\lambda(k,x_{0,j},\theta)$ denotes the costate generated from the terminal condition $\lambda_{T+1}=0$. According to Theorem~\ref{theorem:policy_gradient} with the boundary $S(T+1,x_0,\theta)^\top\lambda_{T+1}=0$, \eqref{eq:truncated_gradient} is the exact gradient of the truncated cost 
\begin{align}\label{eq:cost_truncated}
    \Jcal_T(x_0,\theta):=\sum_{k=0}^{T} l(x(k,x_0,\theta),\theta).
\end{align}
Its distance to $\nabla_\theta\Jcal(x_0,\theta)$ decays exponentially in $T$, as shown in Theorem~\ref{thm:errors}. Accordingly, \eqref{eq:MCsampling} is evaluated with $\nabla_\theta\Jcal_T$ in place of $\nabla_\theta\Jcal$.
 
It follows from Theorem~\ref{theorem:policy_gradient} that computing \eqref{eq:truncated_gradient} requires evaluating the Hamiltonian dynamics \eqref{eq:HamiltonDynamics} over $k=0,\ldots,T$. These recursions are inherently sequential and can become computationally expensive for long horizons. To alleviate this temporal bottleneck, we reformulate the sequential state and costate recursions as a root-finding optimization problem and exploit associative scans \cite{blelloch1989scans} to parallelize the trajectory evaluation across time.

Consider the stacked state and costate estimates
\begin{align*}
    s_{x,j}
    &=
    [\xi_1^\top,\cdots,\xi_T^\top]^\top
    \in\R^{nT},\\
    s_{\lambda,j}
    &=
    [\phi_{T-1}^\top,\cdots,\phi_0^\top]^\top
    \in\R^{nT},
\end{align*}
which approximate the corresponding state and costate trajectories
\begin{align*}
    s_{x,j}^{*}
    &=
    [x(1,x_{0,j},\theta)^\top,\cdots,x(T,x_{0,j},\theta)^\top]^\top
    \in\R^{nT},\\
    s_{\lambda,j}^{*}
    &=
    [\lambda(T-1,x_{0,j},\theta)^\top,\cdots,\lambda(0,x_{0,j},\theta)^\top]^\top
    \in\R^{nT},
\end{align*}
respectively. Hereafter, the subscript $j$, and the policy parameter $\theta$ are omitted for notational simplicity whenever no confusion arises. 
The residuals associated with the stacked state and costate estimates are then defined as
\begin{align*}
    r_x(s_x)
    &=
    [\xi_1^\top-F(\xi_0)^\top,
    \cdots,
    \xi_T^\top-F(\xi_{T-1})^\top]^\top,\\
    r_\lambda(s_\lambda)
    &=
    [\phi_{T-1}^\top-G(x_{T-1},\phi_T)^\top,
    \cdots,
    \phi_0^\top-G(x_0,\phi_1)^\top]^\top.
\end{align*}
 
The two trajectories $s_x^*$ and $s_\lambda^*$ are evaluated using two separate Gauss-Newton passes. The first pass evaluates \(s_x^*\) from a given initial state $\xi_0 = x_{0,j}$. The resulting state trajectory is then held fixed, the terminal condition $\phi_T =\nabla_xl(x_T)$ is evaluated, and the second pass evaluates \(s_\lambda^*\) backward in time through \eqref{eq:lamk1}. In each pass, we solve the following residual-minimization problem for both the state and costate:
\begin{subequations}\label{eq:Deerminimize}
    \begin{align}
    \min_{s_x \in \R^{nT}}\Lcal_x(s_x)=\frac12\|r_x(s_x)\|_2^2 \,\\
    \min_{s_\lambda \in \R^{nT}}\Lcal_\lambda(s_\lambda)=\frac12\|r_\lambda(s_\lambda)\|_2^2 \,
    \end{align}
\end{subequations}
 
The residual Jacobian $J(s_x)=\partial r_x(s_x)/\partial s_x$ and $J(s_\lambda)=\partial r_\lambda(s_\lambda)/\partial s_\lambda$ are square and nonsingular. Consequently, the Gauss-Newton step for \eqref{eq:Deerminimize} is
\begin{subequations}\label{eq:GNIterations}
    \begin{align}
    s_x^{(i+1)}-s_x^{(i)}=\Delta s_x^{(i)}
    =
    -J(s_x^{(i)})^{-1}r_x(s_x^{(i)}), \label{eq:s_update}\\
    s_\lambda^{(i+1)}-s_\lambda^{(i)}=\Delta s_\lambda^{(i)}
    =
    -J(s_\lambda^{(i)})^{-1}r_\lambda(s_\lambda^{(i)}). \label{eq:lambda_update}
\end{align}
\end{subequations}
The Jacobian $J(s_x)$ has a block bidiagonal structure 
\begin{align*}
J(s_x)=
\begin{bmatrix}
I_n
&0
&\cdots
&0 &0\\
-\nabla_xF(\xi_1)
&I_n &\cdots &0 &0\\
\vdots
&\vdots & \ddots &\vdots &\vdots\\
0 &0 &\cdots &I_n &0 \\
0 &0 &\cdots
&-\nabla_xF(\xi_{T-1}) &I_n
\end{bmatrix}.
\end{align*}
 
Similarly, $J(s_\lambda)$ is obtained by replacing $\nabla_xF(\xi_k)$ in the above expression with $\nabla_{\lambda}G(x_{T-k-1},\phi_{T-k})$. Since $\Delta \xi_0^{(i)}=0$, the state correction satisfies
\begin{align}
    \Delta \xi_{k+1}^{(i)}
    =
    \nabla_xF(\xi_k^{(i)},\theta)\Delta \xi_k^{(i)}
    -
    r_{x,{k+1}}^{(i)},
    \label{eq:state_cor}
\end{align}
where $\Delta \xi_{k+1}^{(i)} = \xi_{k+1}^{(i+1)} - \xi_{k+1}^{(i)}$. During the costate pass, the state trajectory and terminal costate are held fixed. Hence, \(\Delta\phi_T^{(i)}=0\), $\Delta \phi_{k}^{(i)} = \phi_{k}^{(i+1)} - \phi_{k}^{(i)}$, and the costate correction satisfies
\begin{align}
    \Delta\phi_k^{(i)}
    =
    \nabla_{\lambda}G(x_k,\phi_{k+1}^{(i)},\theta)
    \Delta\phi_{k+1}^{(i)}
    -
    r_{\lambda,k}^{(i)}.
    \label{eq:costate_cor}
\end{align}
Since $G$ is affine in its second argument, the costate residual $r_\lambda$ is affine in $s_\lambda$. Hence a single step \eqref{eq:lambda_update} from any initialization, returns $s_\lambda^*$ exactly.
 
\subsection{Parallel evaluation by associative scan}
 
Since the state and costate correction equations in  \eqref{eq:state_cor} and \eqref{eq:costate_cor} admit parallel  evaluation through associative scans \cite{blelloch1990prefix}, as illustrated in Figure~\ref{fig:scan}. Consider an affine recursion of the form \(z_{k+1}=E_k z_k+F_k\), where $F_k$ denotes a generic offset and not the closed-loop map $F$. Each update can be viewed as an affine map \(z\mapsto E_k z+F_k\), represented by the pair \((E_k,F_k)\). The composition of two affine maps is represented by the binary operator
\begin{align*}
(E_2,F_2)\odot(E_1,F_1)
:=
\left(E_2E_1,E_2F_1+F_2\right).
\end{align*}
This operator is associative because it corresponds directly to function composition. Consequently, the sequence of affine recursions can be evaluated through a parallel prefix scan, which computes all cumulative affine-map compositions simultaneously. 
 
 
For the state pass, define $\mathfrak a_{x_k}^{(i)}:=(\nabla_xF(\xi_k^{(i)},\theta),-r_{x,{k+1}}^{(i)})$. Then all state corrections are recovered from the prefix scan
\begin{align*}
    \Delta \xi_{k+1}^{(i)}
    =
    \left(
        \mathfrak a_{x_k}^{(i)}
        \odot\cdots\odot
        \mathfrak a_{x_0}^{(i)}
    \right)
    [\Delta \xi_0^{(i)}].
\end{align*}
Similarly, for the fixed-state costate pass, define $\mathfrak a_{\phi_k}^{(i)}:=(\nabla_{\lambda}G(x_k,\phi^{(i)}_{k+1}),-r_{\lambda,k}^{(i)})$, the costate corrections are obtained from the suffix scan
\begin{align*}
    \Delta\phi_k^{(i)}
    =
    \left(
        \mathfrak a_{\phi_k}^{(i)}
        \odot\cdots\odot
        \mathfrak a_{\phi_{T-1}}^{(i)}
    \right)
    [\Delta\phi_T^{(i)}].
\end{align*}
 
\begin{figure*}
    \centering
    \includegraphics[width=0.8\linewidth]{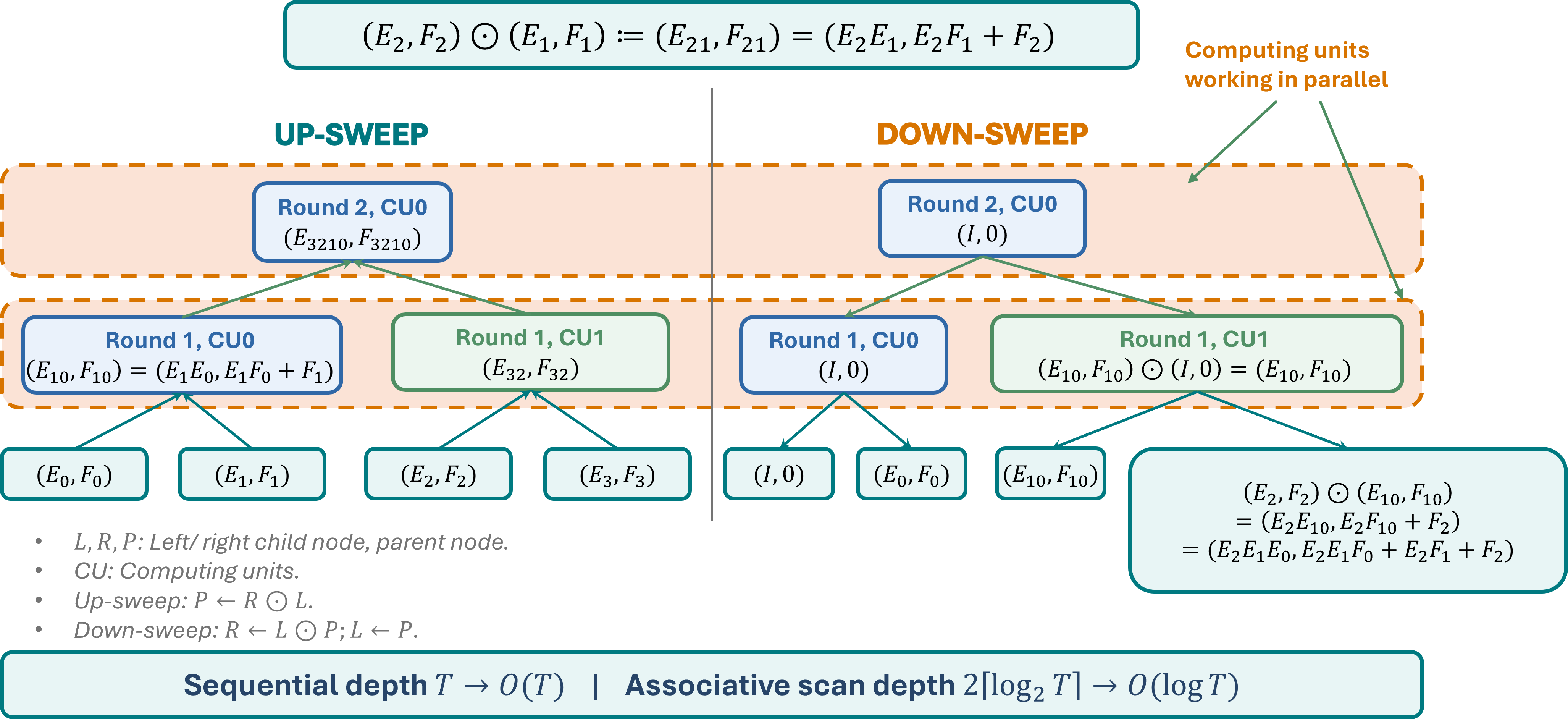}
    \caption{Associative-scan evaluation example of the affine recursion for $T=4$.}
    \label{fig:scan}
\end{figure*}
 
At each Gauss-Newton iteration, residuals and Jacobian blocks are evaluated in parallel across time and combined by an associative scan, yielding \(\mathcal O(\log T)\) parallel depth and \(\mathcal O(T)\) total compositions. For dense Jacobians $\nabla_x F$ and $\nabla_\lambda G$, each composition costs $\mathcal O(n^3)$ operations and each scan element stores an $n\times n$ matrix. Thus, the performance may degrade for dense Jacobians of large dimension, as can also be seen in Table~\ref{tab:lqr_deer_benchmark}. A more memory-efficient variant of the GN method, called quasi-DEER, is given in \cite{gonzalez2024towards}.
 
\begin{algorithm}[!t]
\caption{Monte Carlo Policy-Gradient Optimization}
\label{alg:two_pass_deer}
\begin{algorithmic}[1]
\Require Initial policy parameter $\theta,\theta_{\text{old}}\in\Theta$, horizon $T$, number of Monte Carlo samples $N$, initial-state distribution $\Xcal$, policy step sizes $\{\eta_\ell\}$, tolerance $\epsilon_{\mathrm{D}}$, clipping cap $a$ and number of clipped iterations (both optional), and policy-gradient tolerance $\epsilon_{\mathrm{PG}}$
\Ensure Optimized policy parameter $\theta$
 
\State $\ell \leftarrow 0$
\State Sample $N$ initial states $\{x_{0,j}\}_{j=1}^N
    \overset{\mathrm{i.i.d.}}{\sim}\Xcal$
\Repeat
 
    \ForAll{$j=1,\ldots,N$ \textbf{in parallel}}
        \State  Initialize $s_{x,j}\leftarrow 0$ (or warm start)
        \State Apply the clipped iterations \eqref{eq:clip-update} (optional)
 
        \Repeat
            \State Compute $\Delta s_{x,j}$ from \eqref{eq:state_cor}
            by associative scan
            \State $s_{x,j}
            \leftarrow
            s_{x,j}+\Delta s_{x,j}$
        \Until{$\norm{\Delta s_{x,j}}\le\epsilon_{\mathrm{D}}$}
 
        \State Initialize $s_{\lambda,j}\in \R^{nT}$
 
            \State Compute $\Delta s_{\lambda,j}$ from 
            \eqref{eq:costate_cor}
            by one associative scan
            \State $s_{\lambda,j}\leftarrow s_{\lambda,j}+\Delta s_{\lambda,j}$
        \State Compute $\nabla_\theta\Jcal_T(x_{0,j},\theta)$ from \eqref{eq:truncated_gradient}, $s_{x,j}$ and $s_{\lambda,j}$
    \EndFor
 
    \State Compute $\hat{\nabla}_\theta\bar{\Jcal}(\theta)$ from
    \eqref{eq:MCsampling} with $\nabla_\theta\Jcal_T$
    \State $\theta_{\text{old}}\leftarrow\theta$
    \State $\theta
    \leftarrow
    \mathrm{Proj}_{\Theta}
    \left(
        \theta_{\text{old}}-\eta_\ell\hat{\nabla}_\theta\bar{\Jcal}(\theta_{\text{old}})
    \right)$
    \State $\ell\leftarrow\ell+1$
 
\Until{$\norm{\theta_\text{old}-\theta}
\le\epsilon_{\mathrm{PG}}$}
 
\State \Return $\theta$
\end{algorithmic}
\end{algorithm}
\normalcolor
 
Algorithm~\ref{alg:two_pass_deer} estimates $\nabla_\theta \Bar{\Jcal}=\E_{x_0}\!\left[\nabla_\theta \Jcal(x_0,\theta)\right]$
using Monte Carlo samples. For each initial condition, a forward GN pass computes the state trajectory and a backward pass computes the costate trajectory. The resulting gradients are averaged to update \(\theta\), while the projected gradient descent keeps \(\theta\) in the admissible region $\Theta$.
 
\subsection{Local Convergence Analysis of the GN Iterations}
 
Fix $\theta\in\Theta$ and omit it from $F$ for simplicity. All constants in this subsection are independent of the horizon $T$, of $x_0\in\Omega$, and of $\theta\in\Theta$. Hence the results hold with the same constants for every policy parameter visited by Algorithm~\ref{alg:two_pass_deer}.
Admissibility and compactness of \(\Omega\) and \(\Theta\) imply that the exact trajectories are uniformly bounded over all \(x_0\in\Omega\) and \(\theta\in\Theta\) (see the proof of Lemma~\ref{lem:GN-inverse}).  Let \(K\subset\mathbb{R}^n\) be a compact convex set whose interior contains the closure of the set of all states reached by these trajectories. Define
\[
\|v\|_\infty:=\max_k\|v_k\|_2,
\qquad
e^{(i)}:=s_x^{(i)}-s_x^*.
\]
 
\begin{lemma}
\label{lemma:Lipchitz}
Let $L>0$ be a Lipschitz bound for $\nabla_xF$ on $K$, uniform in $\theta\in\Theta$. Such a constant exists by the smoothness of $F$, the convexity of $K$, and the compactness of $K\times\Theta$. Then
\[
\|J(s_1)-J(s_2)\|_2
\le L\|s_1-s_2\|_\infty,
\qquad \forall s_1,s_2\in K^T.
\]
\end{lemma}
 
\begin{proof}
For $s_1=(z_1,\ldots,z_T)$ and $s_2=(y_1,\ldots,y_T)$,
the Jacobian difference has only the subdiagonal blocks
$\nabla_xF(y_k)-\nabla_xF(z_k)$. Hence, for $T\ge2$,
\begin{align*}
\|J(s_1)-J(s_2)\|_2
&=\max_{1\le k<T}
  \|\nabla_xF(z_k)-\nabla_xF(y_k)\|_2\\
&\le L\|s_1-s_2\|_\infty.
\end{align*}
For \(T=1\), the residual Jacobian is constant, so the bound holds trivially.
\end{proof}

\begin{lemma}
\label{lem:GN-inverse}
There exist $\eta,M>0$, independent of $T$, $x_0\in \Omega$, and $\theta\in\Theta$, such that
\begin{equation}\label{eq:GN-inverse}
\begin{aligned}
\|J(s_x)^{-1}\|_2&\le M,\\
\|J(s_x)^{-1}v\|_\infty&\le M\|v\|_\infty,
\end{aligned}
\end{equation}
for any $s_x \in \R^{nT}$ satisfying $\|s_x-s_x^*\|_\infty\le\eta$ and any $v\in\R^{nT}$.
\end{lemma}

\begin{proof}
Write $F_\theta(x):=F(x,\theta)$ and $A_\theta(x):=\nabla_x F(x,\theta)$. For each $\bar\theta\in\Theta$, LES of the origin implies that $A_{\bar\theta}(0)$ is Schur stable. The discrete Lyapunov inequality therefore yields a matrix $O\succ0$ such that $\|O A_{\bar\theta}(0)O^{-1}\|_2<1$. By continuity of $A_\theta(z)$ in $(z,\theta)$ and compactness of $\Theta$, there exist a finite open cover $\{V_a\}_{a=1}^m$ of $\Theta$, matrices $O_a\succ0$, and constants $c\in(0,1)$ and $r>0$ such that
\[
    \|O_a A_\theta(z)O_a^{-1}\|_2\le c
    \qquad
    \text{for all } \theta\in V_a\cap\Theta,\quad \|z\|_2\le 2r.
\]
Here $c$ is the largest of the finitely many local contraction factors, and $r$ is the smallest of the corresponding radii.

For each $a$, define the ellipsoid
\[
    E_a:=\left\{z:\|O_a z\|_2\le {r}/{\|O_a^{-1}\|_2}\right\}.
\]
Then $E_a\subseteq\{z:\|z\|_2\le r\}$. Since $F_\theta(0)=0$, integration along the segment from $0$ to $z$ gives
\[
    \|O_a F_\theta(z)\|_2\le c\|O_a z\|_2
    \qquad \text{for all } z\in E_a,\quad \theta\in V_a\cap\Theta.
\]
Thus $E_a$ is forward invariant for every such $\theta$.

Fix $(\bar x_0,\bar\theta)\in\Omega\times\Theta$ and choose $a$ such that $\bar\theta\in V_a$. By global attractivity, the corresponding exact trajectory enters $\operatorname{int}E_a$ at some finite time $N$. By continuity of the $N$-step map, all initial state and parameter pairs in a neighborhood of $(\bar x_0,\bar\theta)$ also satisfy $x_N^*\in\operatorname{int}E_a$. Choose this neighborhood small enough that $\theta\in V_a$ throughout. Forward invariance then gives $\|x_k^*\|_2\le r$ for all $k\ge N$. Taking a finite subcover of $\Omega\times\Theta$ and the largest of the associated entry times yields a uniform $\tau$ such that
\[
    \|x_k^*\|_2\le r
    \qquad \text{for all } k\ge\tau,\quad
    (x_0,\theta)\in\Omega\times\Theta.
\]
For $k<\tau$, each finite-step map is continuous on the compact set $\Omega\times\Theta$. Hence the exact states are uniformly bounded over all $k$, as used in the definition of $K$.

Since $\operatorname{int}K$ contains the compact closure of the set of exact states, choose $0<\eta\le r$ such that their closed $\eta$-neighborhood lies in $K$. If $s_x=(\xi_1,\ldots,\xi_T)$ satisfies $\|s_x-s_x^*\|_\infty\le\eta$, then $\xi_k\in K$ for all $k$ and $\|\xi_k\|_2\le r+\eta\le 2r$ for $k\ge\tau$.

Define the uniform constants
\begin{align*}
    \alpha:&=\max\left\{1,\,
    \max_{1\le a\le m}\max_{(z,\theta)\in K\times\Theta}
    \|O_a A_\theta(z)O_a^{-1}\|_2\right\},
    \\
    \kappa:&=\max_{1\le a\le m}\|O_a^{-1}\|_2\|O_a\|_2.
\end{align*}
For any $\theta \in \Theta$, choose $a$ with $\theta\in V_a$. In any product of $\ell$ consecutive Jacobians, at most $\min\{\ell,\tau\}$ factors precede time $\tau$. In the norm induced by $O_a$, these factors are bounded by $\alpha$, and every subsequent factor is bounded by $c$. Using the same $O_a$ throughout the product and converting back to the Euclidean norm gives
\begin{equation}\label{eq:GN-products}
\begin{aligned}
    \|A_\theta(\xi_{j+\ell-1})\cdots A_\theta(\xi_j)\|_2 \le \kappa c^\ell \left({\alpha}/{c}\right)^{\min\{\ell,\tau\}} \le Cc^\ell
\end{aligned}
\end{equation}
with $C:=\kappa\left({\alpha}/{c}\right)^\tau$.

Finally, write $J(s_x)=I-\mathcal N$, where $\mathcal N$ has blocks
$A_\theta(\xi_j)$ on its first subdiagonal and zeros elsewhere.
Since $\mathcal N$ is nilpotent of index at most $T$,
\[
    J(s_x)^{-1}=\sum_{\ell=0}^{T-1}\mathcal N^\ell.
\]
For $\ell\ge1$, $\mathcal N^\ell$ is a block shift whose nonzero
blocks are products of $\ell$ consecutive Jacobians. By
\eqref{eq:GN-products}, each block has spectral norm at most $Cc^\ell$.
Since distinct input blocks are mapped to distinct output blocks,
\[
    \|\mathcal N^\ell\|_2\le Cc^\ell,
    \qquad
    \|\mathcal N^\ell v\|_\infty\le Cc^\ell\|v\|_\infty.
\]
These bounds also hold for $\ell=0$ since $C\ge1$. Summing the
geometric series proves \eqref{eq:GN-inverse} with $M=C/(1-c)$,
independently of $T$, $x_0\in\Omega$, and $\theta\in\Theta$.
\end{proof}

As a direct consequence of Lemma~\ref{lem:GN-inverse}, the following result establishes that the merit function $\Lcal_x$ satisfies a local PL condition.
 
\begin{corollary}
\label{cor:GN-PL}
For $\|s_x-s_x^*\|_\infty\le\eta$, the merit function satisfies
\begin{equation}\label{eq:GN-PL}
\frac12\|\nabla\mathcal L_x(s_x)\|_2^2
\ge M^{-2}\mathcal L_x(s_x).
\end{equation}
\end{corollary}
 
\begin{proof}
Lemma~\ref{lem:GN-inverse} gives
$\sigma_{\min}(J(s_x))\ge M^{-1}$.
Since $\nabla\mathcal L_x(s_x)=J(s_x)^\top r_x(s_x)$,
\[
\frac12\|\nabla\mathcal L_x(s_x)\|_2^2
\ge\frac{1}{2M^2}\|r_x(s_x)\|_2^2
=M^{-2}\mathcal L_x(s_x).
\]
\end{proof}

The following theorem establishes the local linear convergence of the iteration \eqref{eq:s_update}.
 
\begin{theorem}
\label{thm:local_linear}
Define
\begin{equation}\label{eq:GN-basin}
\delta
:=\min\left\{\eta,\frac{1}{LM}\right\}.
\end{equation}
If $\|e^{(0)}\|_\infty\le\delta$, the exact full-step GN iterates of \eqref{eq:s_update} remain in this neighborhood and satisfy
\begin{equation}\label{eq:GN-linear}
\|e^{(i+1)}\|_2\le\frac12\|e^{(i)}\|_2.
\end{equation}
\end{theorem}
 
\begin{proof}
Suppose $\|e^{(i)}\|_\infty\le\delta$.
Applying the fundamental theorem of calculus along the
segment from $s_x^*$ to $s_x^{(i)}$ gives
\[
r_x(s_x^{(i)})
=\int_0^1 J(s_x^*+t e^{(i)})e^{(i)}\,dt,
\]
where $r_x(s_x^*)=0$ was used.
Subtracting $s_x^*$ from the GN update \eqref{eq:s_update}
and substituting this identity yields
\begin{align*}
e^{(i+1)}
&=e^{(i)}-J(s_x^{(i)})^{-1}r_x(s_x^{(i)})\\
&=J(s_x^{(i)})^{-1}
  \int_0^1
  \left[J(s_x^{(i)})-J(s_x^*+t e^{(i)})\right]
  e^{(i)}\,dt.
\end{align*}
 
By Lemma~\ref{lemma:Lipchitz},
\[
\|J(s_x^{(i)})-J(s_x^*+t e^{(i)})\|_2
\le L(1-t)\|e^{(i)}\|_\infty.
\]
Applying this estimate, its blockwise counterpart, and Lemma~\ref{lem:GN-inverse} gives
\begin{align}
\|e^{(i+1)}\|_\infty
&\le\frac{LM}{2}\|e^{(i)}\|_\infty^2,
\label{eq:GN-invariance}\\
\|e^{(i+1)}\|_2
&\le\frac{LM}{2}
     \|e^{(i)}\|_\infty\|e^{(i)}\|_2,
\label{eq:GN-error-bound}
\end{align}
where $\int_0^1(1-t)\,dt=1/2$.
Since $LM\delta\le1$, the first inequality preserves the neighborhood and the second proves \eqref{eq:GN-linear}. Induction completes the proof.
\end{proof}

The following theorem builds the local quadratic convergence rate of \eqref{eq:s_update}.
\begin{theorem}
\label{thm:GN-quadratic}
The neighborhood
$\|s_x-s_x^*\|_\infty\le\delta$,
with $\delta$ given by \eqref{eq:GN-basin},
is a horizon-independent basin of quadratic convergence.
For any initialization in this neighborhood,
\begin{align*}
\|e^{(i+1)}\|_2
&\le\frac{LM}{2}\|e^{(i)}\|_2^2,\\
\|r_x(s_x^{(i+1)})\|_2
&\le\frac{LM^2}{2}
     \|r_x(s_x^{(i)})\|_2^2.
\end{align*}
\end{theorem}
 
\begin{proof}
The error bound follows from \eqref{eq:GN-error-bound}
and $\|e\|_\infty\le\|e\|_2$.
For $h=-J(s)^{-1}r_x(s)$, we have $r_x(s)+J(s)h=0$.
Both iterations and their connecting segment remain in the
neighborhood by Theorem~\ref{thm:local_linear}.
Taylor's formula and Lemma~\ref{lemma:Lipchitz} give
\[
\|r_x(s+h)\|_2
\le\frac L2\|h\|_2^2
\le\frac{LM^2}{2}\|r_x(s)\|_2^2.
\]
\end{proof}
 
 
 
 
 
\subsection{Convergence Analysis of the GN Iterations with Equilibrium Initialization}
\label{subsec:global}
 
Theorem~\ref{thm:local_linear} guarantees convergence when the state
trajectory estimate $s_x^{(i)}$ is sufficiently close to the exact
trajectory $s_x^*$. To establish a global linear convergence rate from
the equilibrium initialization, with a rate bound independent of the
horizon $T$, we clip the Jacobian in the state correction
\eqref{eq:state_cor}. Throughout this subsection, we suppress the
dependence on $\theta$ when it is clear from context. All constants are
uniform over $T\ge 1$, $x_0\in\Omega$, and $\theta\in\Theta$.

For a clipping cap $a>0$, define
\begin{equation*}
B_a(z):=\frac{a}{\max\{a,\|\nabla_xF(z)\|_{\mathrm F}\}}\nabla_xF(z).
\end{equation*}
This rule scales the Jacobian only when its Frobenius norm exceeds $a$ and ensures that $\|B_a(z)\|_2\le\|B_a(z)\|_{\mathrm F}\le a$. Replacing $\nabla_x F(\xi_k^{(i)})$ in \eqref{eq:state_cor} with
$B_a(\xi_k^{(i)})$ gives the clipped GN iteration
\begin{align}\label{eq:clip-update}
\begin{split}
\xi_{k+1}^{(i+1)} &=F(\xi_k^{(i)}) +B_a(\xi_k^{(i)}) \bigl(\xi_k^{(i+1)}-\xi_k^{(i)}\bigr),  \\
\xi_k^{(0)}&=0, \qquad \xi_{0}^{(i)}=x_0, \qquad 0\le k<T.
\end{split}
\end{align}
At each iteration $i\ge 0$, we fix the initial state
$\xi_0^{(i)}=x_0$. We initialize the remaining states at the equilibrium by setting $\xi_k^{(0)}=0$ for $1\le k\le T$. Let $F^j$ denote the $j$-fold composition of $F$. Fix a margin $\bar\delta>0$, and let
\[
\bar M:=\sup_{\theta\in\Theta,\,x_0\in\Omega,\,k\ge0}\|x(k,x_0,\theta)\|_2,
\qquad
D:=\bar M+\bar\delta,
\]
where $\bar M$ is finite by the proof of Lemma~\ref{lem:GN-inverse}.

The following lemma establishes a uniform contraction estimate for an iterate of the closed-loop map, which we use to prove convergence of the clipped GN iteration \eqref{eq:clip-update}.
\begin{lemma}
\label{lem:eventual}
Fix $r>0$. Then $\sup_{\|z\|_2\le r,\,k\ge0}\|F^k(z,\theta)\|_2$ is bounded uniformly in $\theta\in\Theta$. In addition, there exists an integer $b\ge1$, such that
\begin{equation*}
\|F^b(u,\theta)-F^b(v,\theta)\|_2\le\tfrac14\|u-v\|_2  
\end{equation*}
for all $\theta \in \Theta$, and all $u,v \in \mathbb{R}^n$ satisfying $\|u\|_2,\|v\|_2\le r$.

\end{lemma}
\begin{proof}
By admissibility, the origin is LES and globally attractive for every $\theta\in\Theta$. Hence, the argument in the proof of Lemma~\ref{lem:GN-inverse} applies with $\Omega$ replaced by the ball $\{z\in\R^n:\|z\|_2\le r\}$. The trajectories starting in the ball are uniformly bounded.  Applying the Jacobian product estimate \eqref{eq:GN-products} along these exact trajectories gives constants $C_r\ge1$ and $c_r\in(0,1)$, such that, for every $\theta \in \Theta$, $z \in \mathbb{R}^n$ bounded by $r$, and $\ell\ge0$,
\[
\|\nabla_x(F^\ell)(z)\|_2
=\|\nabla_xF(F^{\ell-1}(z))\cdots\nabla_xF(z)\|_2
\le C_rc_r^\ell.
\]
Choose $b\ge1$ such that $C_rc_r^b\le\tfrac14$. Since the ball is convex,
\[
F^b(u)-F^b(v)=\int_0^1\nabla_x(F^b)\bigl(v+t(u-v)\bigr)(u-v)\,dt,
\]
which proves the Lemma.
\end{proof}
 
\begin{theorem}
\label{thm:global}
There exist an integer $b\ge1$, a constant $\bar L\ge1$, and a cap $a_*\in(0,\tfrac12]$, all independent of $T$, of $x_0\in\Omega$, and of $\theta\in\Theta$, such that, for every $0<a\le a_*$, the iteration \eqref{eq:clip-update} satisfies
\begin{equation}\label{eq:global-block}
\|e^{(i+b)}\|_\infty\le\tfrac12\|e^{(i)}\|_\infty,
\qquad i\in\{0,b,2b,\ldots\},
\end{equation}
\begin{equation}\label{eq:global-rate}
\|e^{(i)}\|_\infty
\le(2\bar L+1)^{b-1}\,2^{-\lfloor i/b\rfloor}\|e^{(0)}\|_\infty,
\qquad i\ge0.
\end{equation}
\end{theorem}
\begin{proof}
Note that $e_k^{(\ell)}:=\xi_k^{(\ell)}-x_k$, and suppress the fixed parameter $\theta$ in $F$ and $B_a$.  
\\ \indent
\emph{Choice of uniform constants.}
Set $r:=\bar M+D$ and choose $b$ according to Lemma~\ref{lem:eventual} for this radius. Let $\bar K$ be a compact convex set containing
\[
    \bigcup_{\theta\in\Theta}\ \bigcup_{j=0}^{b}
    \left\{F^j(z,\theta)+w:
    \|z\|_2\le r,\ \|w\|_2\le\bar\delta\right\}.
\]
Such a set exists by compactness and continuity of the finite-step maps. Choose $\bar L\ge1$ such that $\|\nabla_xF(z,\theta)\|_2\le\bar L$ on $\bar K\times\Theta$. By convexity of $\bar K$, each $F(\cdot,\theta)$ is $\bar L$-Lipschitz on $\bar K$. For brevity, set
\begin{align*}
    Z_1 &:=2\bar L+1,
    \qquad Z_2:=2(\bar L+1)bZ_1^{b-1}, \\
    a_* &:=\min\left\{\frac12,\frac{1}{4Z_2},
                     \frac{\bar\delta}{2D Z_2}\right\}.
\end{align*}
These constants are independent of $T$, $x_0$, and $\theta$.
Fix $0<a\le a_*$.
\\ \indent
\emph{One-step error bound.}
Suppose the states of iterate $\ell$ lie in $\bar K$, that is $\xi^{(\ell)}_k \in \bar K$ for all $0 \le k\le T$. The exact states also lie in $\bar K$, since $\|x_k\|_2\le\bar M\le r$. Subtracting $x_{k+1}=F(x_k)$ from \eqref{eq:clip-update} gives
\[
    e_{k+1}^{(\ell+1)}
    =F(\xi_k^{(\ell)})-F(x_k)
     +B_a(\xi_k^{(\ell)})
      \bigl(e_k^{(\ell+1)}-e_k^{(\ell)}\bigr).
\]
Taking the maximum over $k$, using $e_0^{(\ell+1)}=0$ and
$\|B_a(\cdot)\|_2\le a$, yields
\[
    \|e^{(\ell+1)}\|_\infty\le(\bar L+a)\|e^{(\ell)}\|_\infty+a\|e^{(\ell+1)}\|_\infty.
\]
Since $a\le\tfrac12$, we obtain
\begin{equation}\label{eq:one-update}
    \|e^{(\ell+1)}\|_\infty
    \le \frac{\bar L+a}{1-a}\|e^{(\ell)}\|_\infty
    \le Z_1 \|e^{(\ell)}\|_\infty.
\end{equation}
This estimate requires only that iterate $\ell$ lie in $\bar K$.
\\ \indent
\emph{Contraction over one block.}
Fix a block starting at an iteration $i$ for which $\|e^{(i)}\|_\infty\le D$.
Then $\|\xi_k^{(i)}\|_2\le\bar M+D=r$. Starting from this iterate,
define the comparison sequence by
\[
    \widehat\xi_k^{(0)}:=\xi_k^{(i)},\qquad
    \widehat\xi_0^{(j)}:=x_0,\qquad
    \widehat\xi_{k+1}^{(j+1)}:=F(\widehat\xi_k^{(j)}).
\]
Induction on $j$ gives
\[
    \widehat\xi_k^{(j)}
    =
    \begin{cases}
        x_k, & k\le j,\\
        F^j(\xi_{k-j}^{(i)}), & k>j.
    \end{cases}
\]
Thus, for $0\le j\le b$, the closed $\bar\delta$-neighborhood of each $\widehat\xi_k^{(j)}$ lies in $\bar K$. At $j=b$, the comparison error vanishes for $k\le b$. For $k>b$, both $\xi_{k-b}^{(i)}$ and$x_{k-b}$ lie in the ball of radius $r$, so Lemma~\ref{lem:eventual}gives
\[
    \|\widehat\xi_k^{(b)}-x_k\|_2
    =\|F^b(\xi_{k-b}^{(i)})-F^b(x_{k-b})\|_2
    \le\tfrac14\|\xi_{k-b}^{(i)}-x_{k-b}\|_2.
\]
Consequently,
\begin{equation}\label{eq:comp-error}
    \max_{0\le k\le T}\|\widehat\xi_k^{(b)}-x_k\|_2
    \le\tfrac14 \|e^{(i)}\|_\infty.
\end{equation}
\\
It remains to bound the effect of the clipped correction. Define
\[
    d_j:=\max_{0\le k\le T}
    \|\xi_k^{(i+j)}-\widehat\xi_k^{(j)}\|_2,
    \qquad d_0=0.
\]
Fix $1\le j\le b$ and suppose that $\xi_k^{(i+h)}\in\bar K$ for every $0\le h\le j-1$ and $0\le k\le T$. Repeated application of \eqref{eq:one-update} gives $\|e^{(i+h)}\|_\infty\le Z_1^h\|e^{(i)}\|_\infty$ for $0\le h\le j$. Subtracting the comparison update from the clipped update and using the Lipschitz bound and \eqref{eq:one-update} gives
\begin{align*}
    d_j
    &\le \bar L d_{j-1}
         +a\bigl(\|e^{(i+j)}\|_\infty+\|e^{(i+j-1)}\|_\infty\bigr)\\
    &\le \bar L d_{j-1}
         +2a(\bar L+1)Z_1^{j-1}\|e^{(i)}\|_\infty.
\end{align*}
Iterating this inequality from $d_0=0$ and using $\bar L\le Z_1$ yields
\begin{align}\label{eq:distance}
\begin{split}
    d_j &\le 2a(\bar L+1)
        \sum_{h=1}^{j}\bar L^{j-h}Z_1^{h-1}\|e^{(i)}\|_\infty \\
    &\le 2a(\bar L+1)jZ_1^{j-1}\|e^{(i)}\|_\infty \le aZ_2 \|e^{(i)}\|_\infty.
\end{split}
\end{align}
Since $\|e^{(i)}\|_\infty\le D$ and $aZ_2D\le\bar\delta/2$, iterate $i+j$ lies within $\bar\delta/2$ of the comparison sequence and hence lies in $\bar K$. The initial iterate $i$ lies in $\bar K$ because $\|\xi_k^{(i)}\|_2\le r$. Induction therefore establishes
\eqref{eq:distance} throughout the block.
\\ \indent
At $j=b$, the choice $aZ_2\le\tfrac14$ gives $d_b\le \|e^{(i)}\|_\infty/4$. Combining this estimate with \eqref{eq:comp-error} yields
\[
    \|e^{(i+b)}\|_\infty
    \le d_b+\max_{0\le k\le T}\|\widehat\xi_k^{(b)}-x_k\|_2
    \le\tfrac12 \|e^{(i)}\|_\infty.
\]
\\ \indent
\emph{Global rate.}
The equilibrium initialization gives $\|e^{(0)}\|_\infty=\max_{1\le k\le T}\|x_k\|_2\le\bar M<D$. Thus the block estimate applies successively at $i=0,b,2b,\ldots$, proving \eqref{eq:global-block} and giving $\|e^{(mb)}\|_\infty\le2^{-m}\|e^{(0)}\|_\infty$. For an arbitrary iteration $i=mb+j$, with $0\le j<b$, the one-step bound within the block gives
\[
    \|e^{(i)}\|_\infty\le Z_1^j\|e^{(mb)}\|_\infty
    \le Z_1^{b-1}2^{-\lfloor i/b\rfloor}\|e^{(0)}\|_\infty.
\]
Substituting $Z_1=2\bar L+1$ proves \eqref{eq:global-rate}.
\end{proof}

\begin{remark}
\label{rem:global}
Theorem~\ref{thm:global} complements local convergence result in Theorems~\ref{thm:local_linear}. From the equilibrium initialization, the error of \eqref{eq:clip-update} is halved every $b$ iterations, uniformly over all finite horizons $T$. The integer $b$ is of the order of the settling time of the closed loop and the cap $a_*$ is conservative, so the guaranteed rate is slow. This reflects a trade-off: a small cap gives a guarantee that is global and holds for every horizon, including $T\to \infty$, but it discards most of the Jacobian information, so that \eqref{eq:clip-update} approaches the fixed-point iteration $\xi_{k+1}^{(i+1)}=F(\xi_k^{(i)})$, whose rate is set by the closed loop. 
\\
Algorithm~\ref{alg:two_pass_deer} therefore uses clipped steps to enter the region of quadratic convergence and then switches to full GN steps. Since both the initial error bound and the radius of this region are independent of $T$, \eqref{eq:global-rate} gives a horizon-independent upper bound on the number of clipped steps required. 
\end{remark}
 
\subsection{Truncation and Trajectory Errors}
\label{subsec:errors}
 
Algorithm~\ref{alg:two_pass_deer} introduces two approximations in
evaluating the policy gradient. First, it approximates the
infinite-horizon gradient $\nabla_\theta\mathcal{J}$ in
\eqref{eq:policy_gradient} by the finite-horizon gradient
$\nabla_\theta\mathcal{J}_T$ in \eqref{eq:truncated_gradient}.
Second, it uses a state estimate $s_x$ in place of the exact trajectory
$s_x^*$ in both the backward costate recursion and the finite-horizon
gradient formula. The following theorem bounds the gradient error due
to each approximation.

\begin{theorem}\label{thm:errors}
Let $c\in(0,1)$ and $\eta>0$ be as in Lemma~3.2. There exist constants $C_J,C_D>0$, independent of $T$, $x_0\in\Omega$, and $\theta\in\Theta$, such that the following estimates hold.
\\ \indent
\emph{(i) Horizon truncation.} The gradient of the truncated cost satisfies
\[
    \|\nabla_\theta \mathcal{J}(x_0,\theta)-\nabla_\theta \mathcal{J}_T(x_0,\theta)\|_2 \le C_J(T+1)c^{2T}.
\]
\\ \indent
\emph{(ii) Trajectory approximation.} Let
$s_x=(\xi_1,\ldots,\xi_T)$ satisfy $\|s_x-s_x^*\|_2\le\eta$, and set $\xi_0=x_0$. Let $H_T(s_x,\theta)$ denote the value of \eqref{eq:truncated_gradient} obtained by evaluating the costate recursion \eqref{eq:lamk1} along $\{\xi_k\}_{k=0}^T$ with zero terminal costate. Then
\[
    \|H_T(s_x,\theta)-\nabla_\theta \mathcal J_T(x_0,\theta)\|_2
    \le C_D\|s_x-s_x^*\|_2.
\]
\end{theorem}
\begin{proof}
Fix $x_0\in\Omega$ and $\theta\in\Theta$, and suppress $\theta$ in the function arguments. Write
\begin{align*}
    A(x)&:=\nabla_x F(x,\theta),\qquad
    B(x):=\nabla_\theta F(x,\theta),\\
    h(x)&:=\nabla_x l(x,\theta),\qquad
    p(x):=\nabla_\theta l(x,\theta).
\end{align*}
All constants below are independent of $T$, $x_0$, and $\theta$.
\\ \indent
\emph{Uniform decay estimates.}
Let $C\ge1$ and $c\in(0,1)$ be the constants in the Jacobian product bound \eqref{eq:GN-products}. The proof of Lemma~\ref{lem:GN-inverse} also gives a uniform bound on the initial transient, followed by contraction with factor $c$. Hence there exists $C_x>0$ such that
\[
    \|x_k\|_2\le C_xc^k,\qquad k\ge0.
\]
By smoothness and compactness, $A$, $B$, and $h$ admit Lipschitz constants $L_A$, $L_B$, and $L_h$, respectively, on $K$, uniformly in $\theta$. Since $\pi(0,\theta)=0$ for every $\theta$, we have $\nabla_\theta\pi(0,\theta)=0$ and hence $B(0)=0$. Also, $h(0)=0$ because the stage cost has a minimum at the origin. Therefore,
\[
    \|B(x_k)\|_2\le L_B C_xc^k,\qquad
    \|h(x_k)\|_2\le L_h C_xc^k.
\]
Moreover, $p=2(\nabla_\theta\pi)^\top R\pi$ is a product of two smooth factors that vanish at the origin. Subtracting this product at two points and applying the Lipschitz bound for each factor gives a constant $L_p>0$ such that, for all $x,\xi\in K$,
\begin{equation}\label{eq:stage-gradient-difference}
    \|p(\xi)-p(x)\|_2
    \le L_p\bigl(\|x\|_2+\|\xi-x\|_2\bigr)\|\xi-x\|_2.
\end{equation}
In particular, taking $x=0$ and $\xi=x_k$ gives $\|p(x_k)\|_2\le L_p C_x^2c^{2k}$.
\\ \indent
Let $\lambda_k$ denote the finite-horizon costate with
$\lambda_{T+1}=0$, and let $S_k:=S(k,x_0,\theta)$ denote the state
sensitivity with $S_0=0$. Unrolling \eqref{eq:lamk1} backward and
\eqref{eq:sk1} forward, respectively, and applying
\eqref{eq:GN-products} yields
\begin{equation}\label{eq:costate-sensitivity-decay}
\begin{aligned}
    \|\lambda_k\|_2
    &\le C L_h C_x\sum_{j=k}^{T}c^{j-k}c^j
     \le C_\lambda c^k,
     \quad 0\le k\le T,\\
    \|S_k\|_2
    &\le C L_B C_x\sum_{j=0}^{k-1}c^{k-1-j}c^j
     = C L_B C_x\,k c^{k-1}, \, k\ge1,
\end{aligned}
\end{equation}
where $C_\lambda:=C L_h C_x/(1-c^2)$. Extending the first sum
to infinity gives the same bound for the infinite-horizon costate
$\lambda_k^\infty:=\nabla_x\mathcal J(x_k,\theta)$.
\\ \indent
\emph{(i) Horizon truncation.}
For $0\le k\le T$, the finite- and infinite-horizon costates satisfy
the same backward recursion. Their difference at time $k+1$ is
therefore the terminal discrepancy $\lambda_{T+1}^\infty$ propagated
through $T-k$ transposed Jacobians. By \eqref{eq:GN-products},
\[
    \|\lambda_{k+1}^\infty-\lambda_{k+1}\|_2
    \le C c^{T-k}\|\lambda_{T+1}^\infty\|_2
    \le C C_\lambda c^{2T+1-k}.
\]
Multiplying by $\|B(x_k)\|_2\le L_B C_xc^k$ shows that each of the first $T+1$ summands in \eqref{eq:policy_gradient} differs from its counterpart in \eqref{eq:truncated_gradient} by at most $D_1c^{2T+1}$, where $D_1:=L_B C_x C C_\lambda$. For $k>T$, the omitted summand satisfies
\[
    \|p(x_k)+B(x_k)^\top\lambda_{k+1}^\infty\|_2
    \le D_2c^{2k},
    \, D_2:=L_p C_x^2+cL_B C_x C_\lambda.
\]
Summing the first $T+1$ discrepancies and the remaining geometric tail gives
\begin{align*}
    &\|\nabla_\theta\mathcal J(x_0,\theta)
      -\nabla_\theta\mathcal J_T(x_0,\theta)\|_2 \\
    &\quad \le D_1(T+1)c^{2T+1} +\frac{D_2c^{2T+2}}{1-c^2} 
    \le C_J(T+1)c^{2T},
\end{align*}
with $C_J:=cD_1+c^2D_2/(1-c^2)$.
\\ \indent
\emph{(ii) Trajectory approximation.}
Let $\epsilon_k:=\|\xi_k-x_k\|_2$ and $E:=\|s_x-s_x^*\|_2$. Since $\xi_0=x_0$, we have $\epsilon_0=0$ and $\sum_{k=0}^T\epsilon_k^2=E^2$. Let $\phi_k$ denote the costate computed along the approximate trajectory, and set $\zeta_k:=\phi_k-\lambda_k$. Subtracting the two costate recursions gives $\zeta_{T+1}=0$ and
\begin{align*}
    \zeta_k &=A(\xi_k)^\top\zeta_{k+1}+w_k,\\
    w_k &:=h(\xi_k)-h(x_k)
       +[A(\xi_k)-A(x_k)]^\top\lambda_{k+1}.
\end{align*}
The Lipschitz bounds on $h$ and $A$, together with \eqref{eq:costate-sensitivity-decay}, give $\|w_k\|_2\le(L_h+L_A C_\lambda)\epsilon_k$. Since $\|s_x-s_x^*\|_\infty\le E\le\eta$, the product bound \eqref{eq:GN-products} holds along the approximate trajectory. Unrolling the recursion yields
\[
    \|\zeta_k\|_2
    \le C(L_h+L_A C_\lambda)
       \sum_{j=k}^{T}c^{j-k}\epsilon_j.
\]
The geometric kernel has $\ell_1$-norm $(1-c)^{-1}$. Hence Young's convolution inequality gives
\[
    \left(\sum_{k=0}^{T}\|\zeta_k\|_2^2\right)^{1/2}
    \le C_\zeta E,\qquad
    C_\zeta:=\frac{C(L_h+L_A C_\lambda)}{1-c}.
\]
\\ \indent
Finally, the difference between the $k$th summands of $H_T(s_x)$
and $\nabla_\theta\mathcal J_T(x_0,\theta)$ is
\[
    p(\xi_k)-p(x_k)
    +[B(\xi_k)-B(x_k)]^\top\lambda_{k+1}
    +B(\xi_k)^\top\zeta_{k+1}.
\]
Set $\gamma:=(1-c^2)^{-1/2}$, so that
$(\sum_{k=0}^{T}c^{2k})^{1/2}\le\gamma$. Applying
Cauchy--Schwarz to the three terms separately gives
\begin{align*}
    \sum_{k=0}^{T}\|p(\xi_k)-p(x_k)\|_2
    &\le L_p(C_x\gamma+E)E,\\
    \sum_{k=0}^{T}
    \|B(\xi_k)-B(x_k)\|_2\,\|\lambda_{k+1}\|_2
    &\le L_B C_\lambda\gamma E,\\
    \sum_{k=0}^{T}\|B(\xi_k)\|_2\,\|\zeta_{k+1}\|_2
    &\le L_B(C_x\gamma+E)C_\zeta E.
\end{align*}
The first bound follows from \eqref{eq:stage-gradient-difference}.
The last uses $\|B(\xi_k)\|_2\le L_B(C_xc^k+\epsilon_k)$ and
$\zeta_{T+1}=0$. Summing these bounds and using $E\le\eta$ gives
\[
    \|H_T(s_x)-\nabla_\theta\mathcal J_T(x_0,\theta)\|_2
    \le C_D E,
\]
where
\[
    C_D:=(L_p+L_B C_\zeta)(C_x\gamma+\eta)
          +L_B C_\lambda\gamma.
\]
This proves the trajectory approximation bound.
\end{proof}

Let $\epsilon_{\mathrm{D}}\le\eta$. If a state pass is stopped inside the basin of Theorem~\ref{thm:local_linear}, then \eqref{eq:GN-linear} gives $\|e^{(i+1)}\|_2\le\|e^{(i)}\|_2-\|e^{(i+1)}\|_2\le\|\Delta s_x^{(i)}\|_2\le\epsilon_{\mathrm{D}}$ for the returned iterate. By Theorem~\ref{thm:errors}, the direction computed by Algorithm~\ref{alg:two_pass_deer} then differs from the Monte Carlo estimate \eqref{eq:MCsampling} of the exact gradients by at most
\begin{equation}\label{eq:total-error}
C_{\Jcal}\,(T+1)\,c^{2T}+C_{\mathrm{D}}\,\epsilon_{\mathrm{D}}
\end{equation}
in the 2-norm.
 

At each step of Algorithm \ref{alg:two_pass_deer}, it minimizes the Monte Carlo estimate of the finite-horizon cost
\begin{align*}
    \min_{\theta \in \Theta}\bar{\mathcal J}_T(\theta)
    &:=\frac{1}{N}\sum_{j=1}^N\mathcal J_T(x_{0,j},\theta). 
\end{align*}
Without the influence of estimation errors, and under proper assumption on the step size $\eta_\ell$ the algorithm finally converges to the necessary condition
\[
    \liminf_{\ell\to\infty}\|U_{\eta_\ell}(\theta_\ell)\|_2
    = 0.
\]
where 
\begin{align*}
    U_\eta(\theta)
    &:= \left[
        \theta-\operatorname{Proj}_{\Theta}\!\left(
        \theta-\eta\nabla_\theta\bar{\mathcal J}_T(\theta)
        \right)\right]/\eta.
\end{align*}
The next corollary show that the $\liminf_{\ell\to\infty}\|U_{\eta_\ell}(\theta_\ell)\|_2 $ is sufficiently small even under small noise of gradient estimation.
\begin{corollary}
Let $L>0$ be a Lipschitz constant for
$\nabla_\theta\mathcal J_T(x,\cdot)$ on $\Theta$, uniform in
$x\in\Omega$. Suppose that $0<\eta_\ell\le 1/L$ and
$\sum_{\ell=0}^\infty\eta_\ell=\infty$. Assume that, at every
policy update, each state estimate satisfies
$\|s_{x,j}-s_{x,j}^*\|_2\le\epsilon_{\mathrm D}$, where
$\epsilon_{\mathrm D}$ lies within the radius of
Theorem~\ref{thm:errors}. Then
\[
    \liminf_{\ell\to\infty}\|U_{\eta_\ell}(\theta_\ell)\|_2
    \le C_{\mathrm D}\epsilon_{\mathrm D}.
\]
\end{corollary}

\begin{proof}
Averaging the sample gradients preserves the Lipschitz constant $L$.
Let $g_\ell:=\nabla_\theta\bar{\mathcal J}_T(\theta_\ell)$ and $\epsilon_\ell:=\bar H(\theta_\ell)-g_\ell$. By
Theorem~\ref{thm:errors},
\[
    \|\epsilon_\ell\|_2
    \le \frac{C_{\mathrm D}}{N}
        \sum_{j=1}^N\|s_{x,j}-s_{x,j}^*\|_2
    \le C_{\mathrm D}\epsilon_{\mathrm D}.
\]

Define the quadratic model
\[
    Q_\ell(u):=\langle g_\ell,u-\theta_\ell\rangle
        +\frac{1}{2\eta_\ell}\|u-\theta_\ell\|_2^2.
\]
The exact projected step
$z_\ell:=\theta_\ell-\eta_\ell U_{\eta_\ell}(\theta_\ell)$
minimizes $Q_\ell$ over $\Theta$. Since $Q_\ell$ is
$1/\eta_\ell$-strongly convex and $Q_\ell(\theta_\ell)=0$,
\[
    Q_\ell(z_\ell)
    \le -\frac{\|z_\ell-\theta_\ell\|_2^2}{2\eta_\ell}
    =-\frac{\eta_\ell}{2}\|U_{\eta_\ell}(\theta_\ell)\|_2^2.
\]
Similarly, the computed step $\theta_{\ell+1}$ minimizes
$Q_\ell(u)+\langle\epsilon_\ell,u-\theta_\ell\rangle$ over $\Theta$.
Comparing its value at $\theta_{\ell+1}$ and $z_\ell$ gives
\[
    Q_\ell(\theta_{\ell+1})-Q_\ell(z_\ell)
    \le \langle\epsilon_\ell,z_\ell-\theta_{\ell+1}\rangle
       -\frac{\|z_\ell-\theta_{\ell+1}\|_2^2}{2\eta_\ell}
    \le \frac{\eta_\ell}{2}\|\epsilon_\ell\|_2^2.
\]
Since $\eta_\ell\le 1/L$, the descent lemma bounds the change in
$\bar{\mathcal J}_T^N$ by $Q_\ell(\theta_{\ell+1})$. Combining the
preceding estimates yields
\[
    \bar{\mathcal J}_T^N(\theta_{\ell+1})
    -\bar{\mathcal J}_T^N(\theta_\ell)
    \le -\frac{\eta_\ell}{2}
    \left(\|U_{\eta_\ell}(\theta_\ell)\|_2^2
          -(C_{\mathrm D}\epsilon_{\mathrm D})^2\right).
\]

If the conclusion failed, the expression in parentheses would be
bounded below by some $\gamma>0$ for all sufficiently large $\ell$.
Summing the inequality and using $\sum_{\ell=0}^\infty\eta_\ell=\infty$
would then give $\bar{\mathcal J}_T(\theta_\ell)\to-\infty$,
contradicting the lower boundedness of $\bar{\mathcal J}_T$ on
$\Theta$.
\end{proof}

By nonexpansiveness of the projection, Theorem~\ref{thm:errors}, and the preceding corollary, the infinite-horizon empirical cost satisfies
\begin{align*}
&\liminf_{\ell\to\infty}
\frac{\left\|\theta_\ell-\operatorname{Proj}_{\Theta}\!\left(
\theta_\ell-\eta_\ell\nabla_\theta\bar{\mathcal J}(\theta_\ell)
\right)\right\|_2}{\eta_\ell} \\
&\le C_{\mathrm D}\epsilon_{\mathrm D}+C_J(T+1)c^{2T}.
\end{align*}

\begin{remark}
\label{rem:init}
The equilibrium initialization is independent of $T$. A warm start instead uses the trajectory computed at the previous policy iteration. By the bound \eqref{eq:costate-sensitivity-decay} on $S$, the exact trajectories of $\theta,\theta_{\text{old}}\in\Theta$ differ in $\|\cdot\|_\infty$ by at most a constant times $\|\theta-\theta_{\text{old}}\|_2$. By the argument leading to \eqref{eq:GN-products}, those of two sufficiently close initial states $x_0,x_0'$ differ by at most a constant times $\|x_0-x_0'\|_2$. Hence, if the initial states are kept fixed or change only slightly between policy iterations, and $\eta_\ell$ and $\epsilon_{\mathrm{D}}$ are sufficiently small, the warm start lies inside the basin of Theorem~\ref{thm:local_linear} for every $T$, and the state pass converges quadratically. Otherwise, the equilibrium initialization is used.
\end{remark}
\section{Numerical Experiments}
\begin{figure}[!htp]
\centering
\includegraphics[width=1\linewidth]{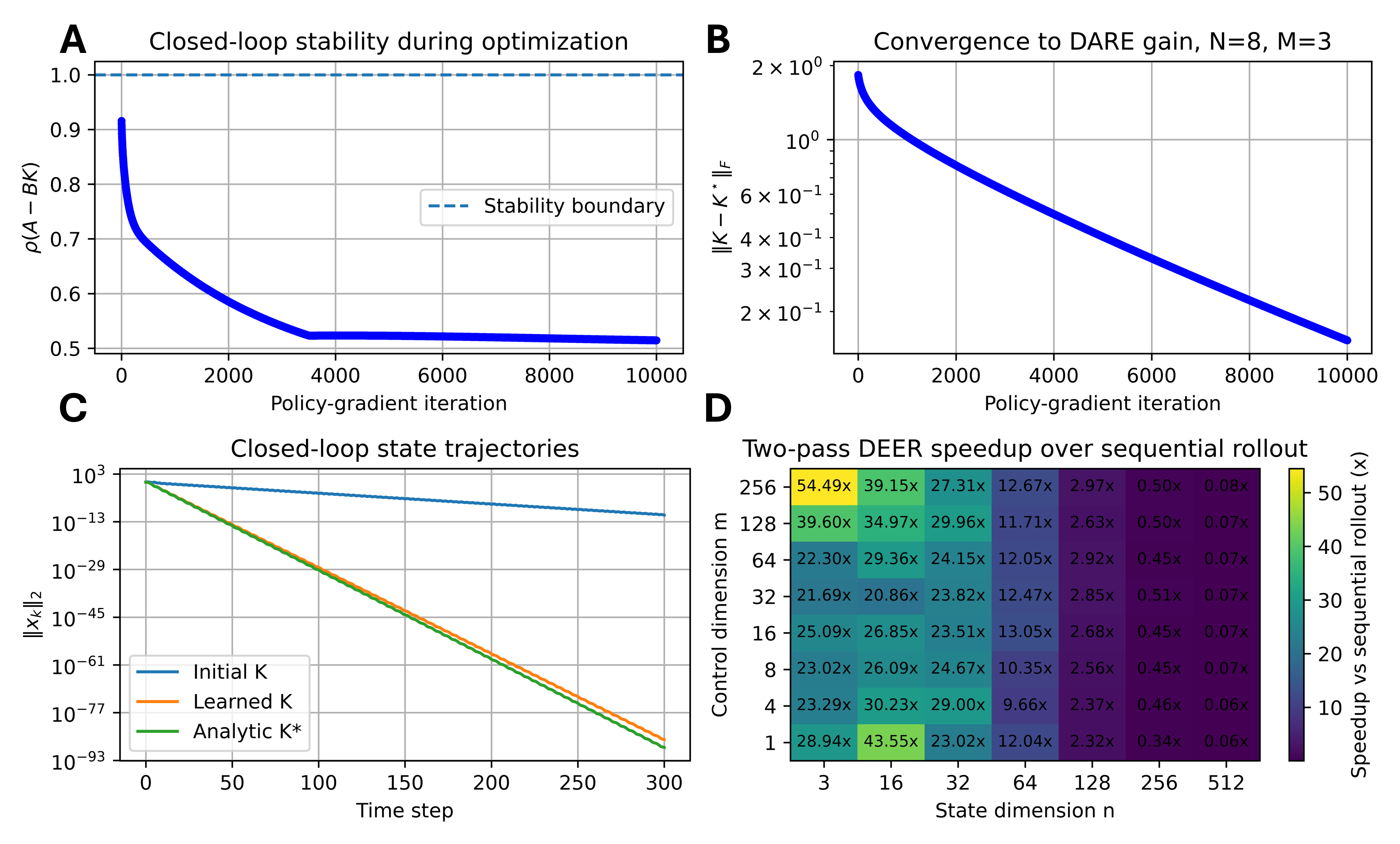}
\caption{LQR validation and computational performance. A: Stability of LQR system under the controller during optimization. B: Difference between the optimized gain and analytic optimal gain through time. C: Trajectory errors through time with the initial, learned, and analytical gain. D: Performance speedup when using the method comparing to sequential roll out.}
\label{fig:panel1}
\end{figure}
All experiments are run on a workstation equipped with an Intel Core Ultra 9 285K CPU, 128 GB of RAM, and an NVIDIA GeForce RTX 5090 GPU with 32 GB of memory. We use JAX \cite{jax2018github} throughout the experiments. The GN passes follow the implementation of DEER given in \cite{gonzalez2026predictability}, extended by the clipped iteration \eqref{eq:clip-update}. All computations use double precision, all reported times exclude compilation, and the settings are listed in Table~\ref{tab:settings}. All details for reproducing the experiments are available at \url{https://github.com/lc-lab25/Parallel-Policy-Gradient.git}.
\subsection{Discrete-time LQR}
\label{subsec:lqr}
\begin{table}[!t]
\centering

\caption{Experimental settings.}
\label{tab:settings}
\setlength{\tabcolsep}{3pt}
\begin{tabular}{@{}lll@{}}
\toprule
 & LQR & Lorenz-96\\
\midrule
Dimensions & $n=8$, $m=3$ & $n=m=16$\\
Parameters $\theta$ & gain $K$ & $\kappa_i,\bar u_i$, $i=1,\ldots,n$\\
$q(x)$, $R$ & $x^\top x$, $I_m$ & $hx^\top x$, $hrI_n$, $r=0.1$\\
Sampling period $h$ & -- & $0.01$\\
Step size $\eta_\ell$ & $10^{-2}$ & Armijo rule~\cite{bertsekas1982projected}\\
Tolerance $\epsilon_{\mathrm{D}}$ & $10^{-9}$ & $10^{-8}$\\
Tolerance $\epsilon_{\mathrm{PG}}$ & -- & $10^{-4}$\\
\bottomrule
\end{tabular}
\end{table}
The LQR experiment verifies the method where the optimal gain is known analytically from the discrete-time algebraic Riccati equation; the state residual is affine, so the state pass returns the exact trajectory in one GN iteration. We consider a randomly generated system matrix $A$, scaled to the spectral radius $0.85$, with Gaussian entries of $B$, and choose an initial feedback controller $u=-Kx$ such that the resulting closed-loop system is stable. The stage cost is taken as $l=x^\top x+u^\top u$, and the controller is updated according to the policy-gradient expression in \eqref{eq:lqr_policy_gradient}, without projection. Figure~\ref{fig:panel1}A--C show that the closed loop remains stable at every iteration of Algorithm~\ref{alg:two_pass_deer}, that the gain converges to the analytical optimal LQR gain, and the resulting trajectory errors. In Figure~\ref{fig:panel1}D and Table~\ref{tab:lqr_deer_benchmark} the parallel solver and a sequential rollout at $T=1000$ agree to $10^{-14}$, with a speedup of $26\times$ for small $n$ that is lost for $n\ge2^8$. Each reported result is averaged over $20$ independent runs.
\subsection{Lossless Lorenz-96 System with Saturated IDA-PBC}
\label{subsec:lorenz}
\subsubsection{System and Discrete-Time Model}
The Lorenz-96 model \cite{lorenz1996predictability} is a nonlinear benchmark system of free dimension $n\ge4$ whose quadratic part preserves the energy $V(x):=\frac12\|x\|_2^2$. We drop its linear damping and constant forcing and take one input per state, so that with $x=[z_1\ \cdots\ z_n]^\top\in\R^n$ and $u\in\R^n$,
\begin{equation}\label{eq:l96_ct}
\begin{gathered}
\dot x=W(x)\,x+u=W(x)\,\nabla V(x)+u,\\
\setlength{\arraycolsep}{3.5pt}
W(x):=\begin{bmatrix}
0 & z_n & & & -z_{n-1}\\
-z_n & 0 & z_1 & & \\
 & -z_1 & 0 & \ddots & \\
 & & \ddots & \ddots & z_{n-2}\\
z_{n-1} & & & -z_{n-2} & 0
\end{bmatrix},
\end{gathered}
\end{equation}
where the blank entries are zero and $W(x)=-W(x)^\top$. This is a port-Hamiltonian system without dissipation, $\dot V=x^\top u$, so the uncontrolled system is not asymptotically stable.
\begin{table}[!tp]
\centering
\caption{Sequential versus parallel two-pass Gauss-Newton method for the discrete-time LQR, $T=1000$.}
\label{tab:lqr_deer_benchmark}
\begin{tabular}{rrrrrrrrrr}
\toprule
$n\times m$
& \multicolumn{1}{c}{Seq. (ms)}
& \multicolumn{1}{c}{DEER (ms)}
& \multicolumn{1}{c}{Speedup}
& \multicolumn{1}{c}{$\|x_{\mathrm{deer}}-x_{\mathrm{seq}}\|$}\\
\midrule
3$\times$1 & $15.5 \pm 0.4$  & $0.6 \pm 0.1$   & $26\times$ & $1.1{\times}10^{-15}$\\
16$\times$8 & $13.4 \pm 0.2$  & $0.6 \pm 0.1$   & $22\times$ & $1.3{\times}10^{-15}$\\
32$\times$16 & $13.2 \pm 0.2$  & $0.5 \pm 0.1$   & $24\times$ & $2.9{\times}10^{-15}$\\
64$\times$32 & $13.4 \pm 0.3$  & $1.2 \pm 0.1$   & $11\times$ & $3.0{\times}10^{-15}$\\
$2^7\times$64 & $18.7 \pm 0.3$  & $5.7 \pm 0.2$   & $3\times$  & $5.6{\times}10^{-15}$\\
$2^8\times2^7$ & $19.6 \pm 0.1$  & $39.3 \pm 0.2$  & $0.5\times$  & $7.1{\times}10^{-15}$\\
$2^9\times2^8$ & $23.9 \pm 0.3$  & $293.3 \pm 3.6$ & $0.08\times$  & $1.1{\times}10^{-14}$\\
\bottomrule
\end{tabular}
\end{table}
With a sampling period $h>0$ and a hold on the input, a forward Euler step on the input and a midpoint step with $W$ frozen at $x_k$ on the lossless part, $x_k'=x_k+hu_k$ and $x_{k+1}=x_k'+\frac h2W(x_k)(x_k'+x_{k+1})$, give the discrete-time model
\begin{equation}\label{eq:l96_dt}
\begin{aligned}
x_{k+1}&=U(x_k)\,(x_k+hu_k),\\
U(x)&:=\bigl(I_n-\tfrac h2W(x)\bigr)^{-1}\bigl(I_n+\tfrac h2W(x)\bigr),
\end{aligned}
\end{equation}
with $f(x)=U(x)x$ and $g(x)=hU(x)$. As the Cayley transform of the skew-symmetric matrix $\frac h2W(x)$, the matrix $U(x)$ is well defined, smooth, and orthogonal \cite{hairer2006geometric}.
\subsubsection{IDA-PBC Controller and Parameter Set}
Following IDA-PBC \cite{ortega2002interconnection}, the input preserves $W(x)$ and the energy $\frac12x^\top x$ and assigns a damping matrix $R_d(x,\theta)\succ0$, $\dot x=(W(x)-R_d(x,\theta))x$. Since every state is actuated, this is achieved by $u=-R_d(x,\theta)x$ with the diagonal choice
\[
R_d(x,\theta):=\mathrm{diag}\Bigl(\kappa_1\psi\Bigl(\frac{\kappa_1z_1}{\bar u_1}\Bigr),\ldots,\kappa_n\psi\Bigl(\frac{\kappa_nz_n}{\bar u_n}\Bigr)\Bigr),
\]
where $\psi(s):=\tanh(s)/s$, $\psi(0):=1$, and $\theta=[\kappa_1\ \cdots\ \kappa_n\ \ \bar u_1\ \cdots\ \bar u_n]^\top$. The policy is
\begin{equation}\label{eq:l96_policy}
\pi(x,\theta)=-R_d(x,\theta)\,x
=-\begin{bmatrix}
\bar u_1\tanh(\kappa_1z_1/\bar u_1)\\
\vdots\\
\bar u_n\tanh(\kappa_nz_n/\bar u_n)
\end{bmatrix},
\end{equation}
with $\kappa_i$ the gain and $\bar u_i$ the saturation level of the $i$-th actuator. It is smooth, satisfies $\pi(0,\theta)=0$, and is nonlinear in the state and in the parameters: the $i$-th input depends only on $z_i$, through a bounded strictly monotone function of slope $-\kappa_i$ at the origin and bound $\bar u_i$, which approaches the relay $-\bar u_i\,\mathrm{sign}(z_i)$ for large $\kappa_i$. By \eqref{eq:l96_dt}, the closed loop \eqref{eq:closedloop} is
\begin{equation}\label{eq:l96_cl}
F(x,\theta)=U(x)\bigl(I_n-hR_d(x,\theta)\bigr)x.
\end{equation}

In continuous time, $\frac{d}{dt}\frac12x^\top x=-x^\top R_d(x,\theta)x<0$ for $x\neq0$ whenever $\kappa_i,\bar u_i>0$, so the design admits all positive gains and saturation levels. In discrete time, the Jacobian of \eqref{eq:l96_cl} at the origin is $I_n-h\,\mathrm{diag}(\kappa_1,\ldots,\kappa_n)$, which is Schur stable if and only if $0<h\kappa_i<2$. We therefore restrict the gains to $h\kappa_i\le1$ and take the compact convex parameter set
\begin{equation}\label{eq:l96_Theta}
\Theta:=\Bigl\{\theta:\ \begin{array}{l}\kappa_{\min}\le\kappa_i\le\kappa_{\max},\\ u_{\min}\le\bar u_i\le u_{\max},\end{array}\ i=1,\ldots,n\Bigr\},
\end{equation}
with $0<\kappa_{\min}\le\kappa_{\max}\le1/h$ and $0<u_{\min}\le u_{\max}$, for which $\mathrm{Proj}_\Theta$ clips every parameter to its interval.
\begin{proposition}
\label{prop:l96}
For every $\theta\in\Theta$ and $x\in\R^n$, the closed loop \eqref{eq:l96_cl} satisfies
\begin{equation}\label{eq:l96_bound}
\begin{aligned}
\|F(x,\theta)\|_2&\le\Bigl(1-h\kappa_{\min}\,\psi\Bigl(\frac{\kappa_{\max}\|x\|_2}{u_{\min}}\Bigr)\Bigr)\|x\|_2,\\
\|F(x,\theta)\|_2&\ge\|x\|_2-h\sqrt n\,u_{\max}.
\end{aligned}
\end{equation}
Consequently, the origin is GAS and LES for every $\theta\in\Theta$, and the origin is not globally exponentially stable.
\end{proposition}
\begin{proof}
Since $U(x)$ is orthogonal, $\|F(x,\theta)\|_2=\|(I_n-hR_d(x,\theta))x\|_2$, and $I_n-hR_d(x,\theta)$ is diagonal with entries $1-h\kappa_i\psi(\kappa_iz_i/\bar u_i)$. As $\psi$ is even, decreasing on $[0,\infty)$, with values in $(0,1]$, and $|\kappa_iz_i/\bar u_i|\le\kappa_{\max}\|x\|_2/u_{\min}$, together with $h\kappa_{\max}\le1$ every entry lies between $0$ and $1-h\kappa_{\min}\psi(\kappa_{\max}\|x\|_2/u_{\min})$, which gives the first bound; the second follows from $\|R_d(x,\theta)x\|_2=\|\pi(x,\theta)\|_2\le\sqrt n\,u_{\max}$.

By the first bound $\|x_k\|_2$ is nonincreasing, hence
\[
\|x_k\|_2\le\Bigl(1-h\kappa_{\min}\,\psi\Bigl(\frac{\kappa_{\max}\|x_0\|_2}{u_{\min}}\Bigr)\Bigr)^{k}\|x_0\|_2,\qquad k\ge0,
\]
so the origin is stable, attracts every trajectory, and is exponentially stable on every bounded set, which is GAS and LES. By the second bound, $\|x_k\|_2\ge\|x_0\|_2-kh\sqrt n\,u_{\max}$, which is incompatible with $\|x_k\|_2\le C\lambda^k\|x_0\|_2$ for all $x_0\in\R^n$ with $C>0$ and $\lambda\in(0,1)$.
\end{proof}
The bounds \eqref{eq:l96_bound} hold uniformly on $\Theta$, so the policy \eqref{eq:l96_policy} is admissible on \eqref{eq:l96_Theta} and the analysis of Section~III applies. The example is GAS and LES without being globally exponentially stable, so the closed loop is not a global contraction: exactly the case covered by the admissibility requirement of Section~II.
\subsubsection{Experiments}
The stage cost is $l(x,\theta)=h\bigl(x^\top x+r\,\pi(x,\theta)^\top\pi(x,\theta)\bigr)$, the bounds of \eqref{eq:l96_Theta} are $\kappa_{\min}=0.5$, $\kappa_{\max}=1/h$, $u_{\min}=0.5$, $u_{\max}=3$, and the policy is initialized at $\kappa_i=\bar u_i=1$. The state pass is started from the equilibrium $\xi_k^{(0)}=0$, from a random guess $\xi_k^{(0)}=x_0+w_k$ with independent Gaussian $w_k$ of standard deviation $0.05$ followed by $100$ clipped iterations \eqref{eq:clip-update} with the cap $a=0.7\sqrt n$, and from a warm start. Two baselines compute the same truncated gradient \eqref{eq:truncated_gradient} on the same GPU: a sequential evaluation of \eqref{eq:closedloop} and \eqref{eq:lamk1} with \texttt{jax.lax.scan}, and BPTT through the rollout. Times are averaged over $20$ runs and memory is the peak GPU allocation.

Figure~\ref{fig:l96_policy} shows the closed-loop states before and after the optimization, the cost over the policy iterations, and the optimized parameters. The cost decreases from $16.74\pm3.73$ to $7.58\pm1.51$, and from $15.81\pm4.08$ to $7.24\pm1.66$ on $1000$ initial states not used in the optimization. The saturation levels converge to $u_{\max}$, where the projection is active, and the gains to interior values between $4.08$ and $4.35$.

Table~\ref{tab:l96_init} compares the three initial guesses of the state pass and the two baselines at the default setting. All five reach the same cost, with final parameters differing by less than $4\times10^{-15}$: the choice of the guess changes the cost of the state pass, not the policy iterates. The warm start is the cheapest, and the random guess the most expensive because of its $100$ clipped iterations.

Figure~\ref{fig:l96_solver}A shows the error of the state pass for the three guesses. The random guess stays near $x_0$, where the spectral radius of $\nabla_xF$ reaches $1.015$, so the Jacobian products grow to $1.5\times10^{13}$ at $T=2000$ instead of satisfying \eqref{eq:GN-products} and the Gauss-Newton iterates overflow within two iterations; the clipped iteration reduces the error linearly, from $5.4$ to $0.31$, after which Gauss-Newton converges. The equilibrium guess leaves the basin of Theorem~\ref{thm:local_linear} as well, for initial states in the larger box $|z_i|\le4$ and for one sample at $n\ge64$, where the same safeguard restores convergence at the price of many iterations (Remark~\ref{rem:global}).

Figures~\ref{fig:l96_solver}B and~\ref{fig:l96_solver}C and Table~\ref{tab:l96_scaling} report the scaling in $T$, $n$, and $N$. The iteration counts of the state pass are independent of the horizon and grow only mildly with $n$ and $N$, in accordance with Theorems~\ref{thm:local_linear} and~\ref{thm:global}, whereas the random guess without clipping converges only at $T=250$. The speedup over the sequential rollout is horizon-independent but decreases with $n$ and $N$, since the baselines do not fill the GPU while a scan costs $\mathcal{O}(NTn^3)$ and stores $\mathcal{O}(NTn^2)$: it is lost for $n\ge64$ and $N=256$, where BPTT is fastest. The memory of the proposed method stays below the $20$ GB budget by chunking the samples, and BPTT exceeds the GPU memory at $n=128$. A diagonal quasi-Newton approximation of the Jacobians would reduce both to $\mathcal{O}(NTn)$ \cite{gonzalez2024towards}, but is not covered by Section~III.

Figure~\ref{fig:l96_solver}D shows the gradient lost by truncating the horizon. The error decays exponentially in $T$, as stated in Theorem~\ref{thm:errors}, and faster at the optimized parameters; at $T=2000$ doubling the horizon changes the gradient by a relative $5\times10^{-16}$, and the gradients of the proposed method and of the two baselines agree to $7\times10^{-16}$.
\begin{figure}[!t]
\centering
\includegraphics[width=1\linewidth]{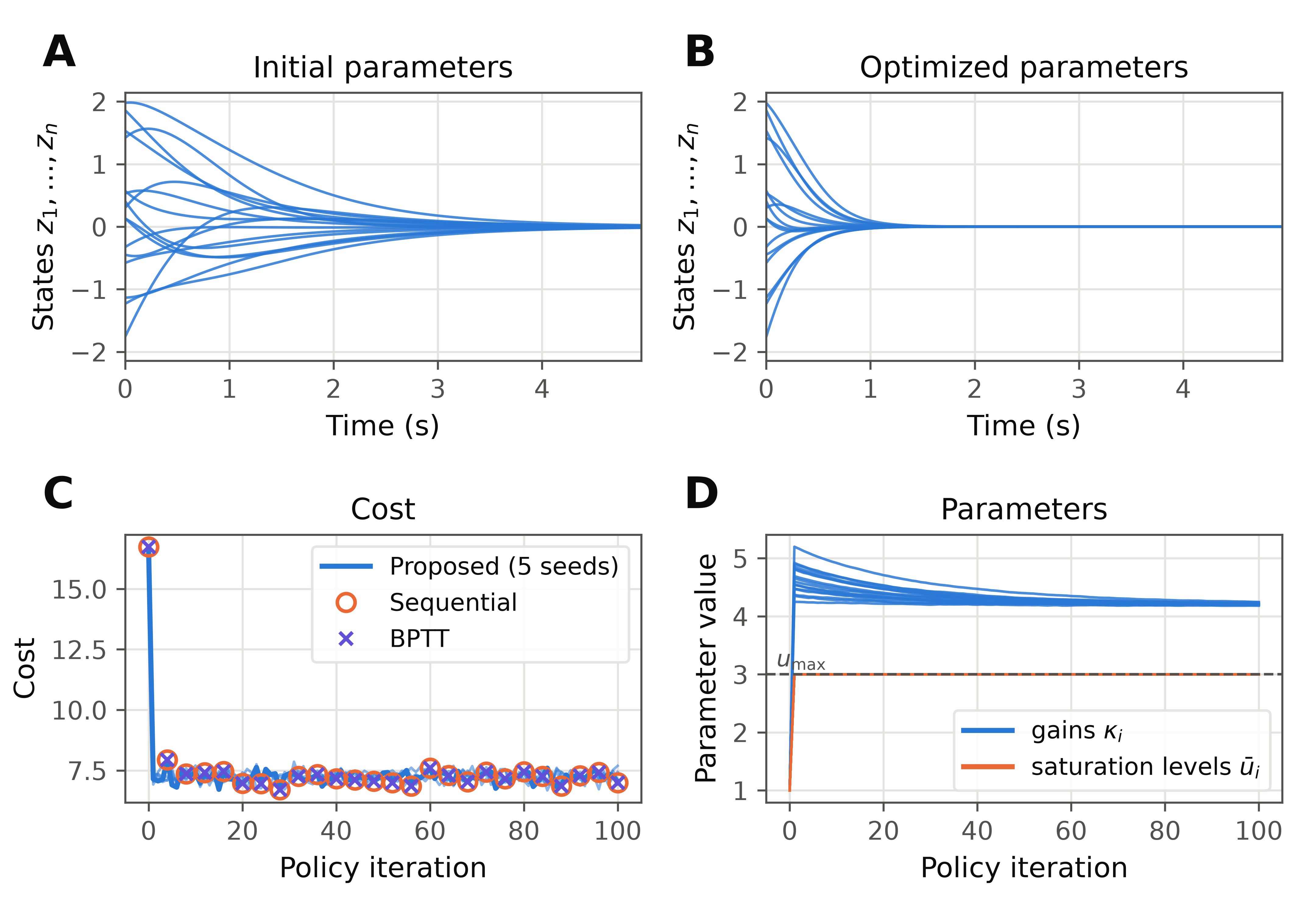}
\caption{Policy optimization for the Lorenz-96 system. A, B: Closed-loop states for one initial state with the initial and the optimized parameters. C: Cost over the policy iterations, the markers show the two baselines for the first set. D: Gains $\kappa_i$ and saturation levels $\bar u_i$ for the first set.}
\label{fig:l96_policy}
\end{figure}
\begin{figure}[!t]
\centering
\includegraphics[width=1\linewidth]{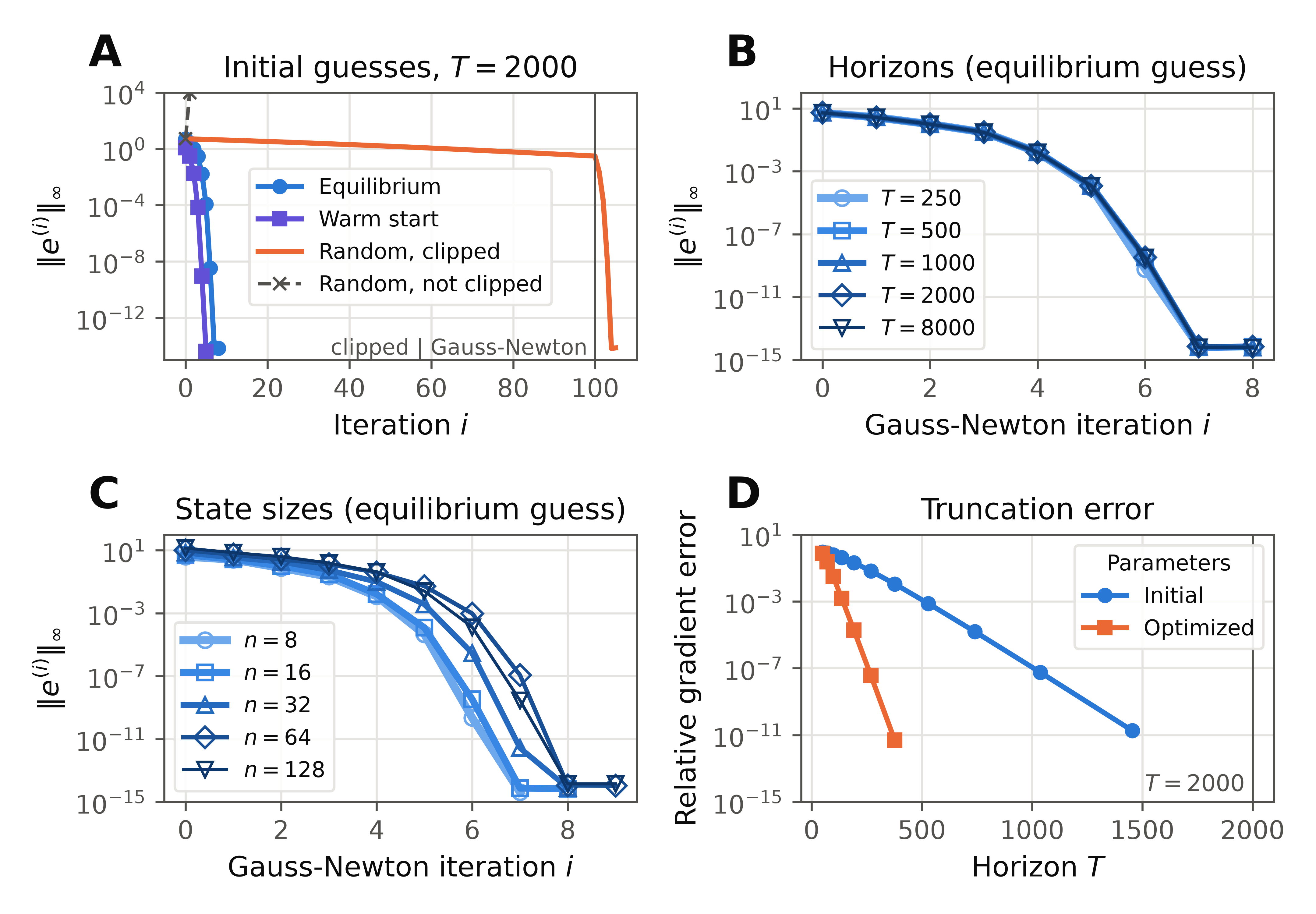}
\caption{State pass and truncation error. A: Error $\|e^{(i)}\|_\infty$, measured against the sequential rollout, for the three initial guesses, the vertical line marks the end of the clipped iterations. B, C: Error for the equilibrium guess for different horizons and state sizes. D: Policy gradient lost by stopping the sum at a finite horizon $T$ instead of $T=8000$, at the initial and the optimized parameters.}
\label{fig:l96_solver}
\end{figure}
\begin{table}[!t]
\centering

\caption{State pass from three initial guesses and the two baselines ($n=16$, $T=2000$, $N=64$, 100 policy iterations): iterations per pass (mean, maximum), final cost, and total time.}
\label{tab:l96_init}
\setlength{\tabcolsep}{3pt}
\begin{tabular}{@{}lccccc@{}}
\toprule
 & Clip. it. & GN it. & Cost & Time (s)\\
\midrule
Equilibrium & 0 & 6.2 (9) & 7.5834 & 16.48\\
Random around $x_0$ & 100 & 2.4 (5) & 7.5834 & 138.12\\
Warm start & 0 & 3.2 (6) & 7.5834 & 12.31\\
\midrule
Sequential & -- & -- & 7.5834 & 37.28\\
BPTT & -- & -- & 7.5834 & 48.68\\
\bottomrule
\end{tabular}
\end{table}
\begin{table}[!t]
\centering

\caption{Scaling in $T$, $n$, and $N$ (unless varied, $n=16$, $T=2000$, $N=64$): Gauss-Newton iterations of the state pass from the equilibrium guess / from a warm start, policy iterations (PG), time of one policy-gradient evaluation in ms for the proposed method (Par.), the sequential evaluation (Seq.), and BPTT, speedup of Par.\ over Seq., and peak memory in GB of Par./Seq./BPTT. An entry $a{+}b$ means that $a$ clipped iterations preceded the Gauss-Newton steps.}
\label{tab:l96_scaling}
\setlength{\tabcolsep}{2.5pt}
\begin{tabular}{@{}lccccccc@{}}
\toprule
 & GN & PG & Par. & Seq. & BPTT & Speedup & Mem.\\
\midrule
$T=250$ & 8/5 & 96 & \textbf{9.6} & 45.9 & 36.6 & 4.8$\times$ & 0.17/0.07/0.12\\
$T=500$ & 9/6 & 96 & \textbf{22.3} & 86.2 & 71.9 & 3.9$\times$ & 0.29/0.02/0.23\\
$T=1000$ & 9/6 & 96 & \textbf{50.8} & 182 & 142 & 3.6$\times$ & 0.59/0.03/0.46\\
$T=2000$ & 9/6 & 96 & \textbf{98.4} & 343 & 286 & 3.5$\times$ & 1.2/0.09/0.93\\
$T=4000$ & 9/6 & 96 & \textbf{196} & 682 & 560 & 3.5$\times$ & 2.4/0.10/1.9\\
$T=8000$ & 9/6 & 96 & \textbf{400} & 1362 & 1146 & 3.4$\times$ & 4.8/0.20/3.7\\
\midrule
$n=8$ & 9/6 & 93 & \textbf{26.8} & 252 & 218 & 9.4$\times$ & 0.32/0.08/0.27\\
$n=16$ & 9/6 & 96 & \textbf{98.4} & 343 & 286 & 3.5$\times$ & 1.2/0.09/0.93\\
$n=32$ & 9/6 & 92 & \textbf{396} & 711 & 526 & 1.8$\times$ & 4.5/0.10/3.4\\
$n=64$ & 10+4/6 & 93 & 1901 & 1426 & \textbf{971} & 0.7$\times$ & 8.3/0.20/13.1\\
$n=128$ & 10+5/6 & 93 & 9654 & \textbf{3550} & -- & 0.4$\times$ & 8.2/0.39/--\\
\midrule
$N=1$ & 7/5 & -- & \textbf{2.3} & 377 & 270 & 164.3$\times$ & 0.10/0.07/0.03\\
$N=16$ & 8/6 & 97 & \textbf{24.4} & 360 & 273 & 14.7$\times$ & 0.23/0.02/0.23\\
$N=64$ & 9/6 & 96 & \textbf{98.4} & 343 & 286 & 3.5$\times$ & 0.91/0.09/0.93\\
$N=256$ & 10/6 & 91 & 403 & 384 & \textbf{284} & 0.95$\times$ & 3.6/0.20/3.7\\
\bottomrule
\end{tabular}
\end{table}
\section{CONCLUSIONS}

This letter develops a time-parallel policy-optimization framework for discrete-time nonlinear control-affine systems. The policy gradient is expressed in terms of forward state and backward costate trajectories, which are evaluated in parallel using Gauss-Newton iterations and associative scans. We establish that the Gauss-Newton state solver converges locally at a horizon-independent linear rate, which complements its known termination within \(T\) iterations \cite{gonzalez2024towards}. A clipped variant converges globally at a horizon-independent linear rate from the equilibrium initialization, and the policy-gradient errors caused by truncating the horizon and by the solver tolerance are bounded explicitly. Owing to its linear structure, the costate rollout is recovered exactly in a single Gauss-Newton iteration once the state trajectory is fixed. Numerical results for LQR and a lossless Lorenz-96 system with saturated IDA-PBC demonstrate the effectiveness and computational efficiency of the proposed approach. Future work will investigate continuous-time extensions and time-parallel sampling-based control methods.


\bibliographystyle{IEEEtran}
\bibliography{ref}

\end{document}